\documentclass[journal]{IEEEtran}

\usepackage[english]{babel}
\usepackage{indentfirst}
\usepackage{amssymb,amsmath}
\usepackage{multirow}
\usepackage{array}
\usepackage[dvipsnames]{xcolor}
\usepackage{mathptmx} 
\usepackage{url}
\usepackage{tikz}
\usetikzlibrary{calc}
\usepackage{float}
\usepackage{soul}
\usetikzlibrary{positioning,arrows,shapes,calc} 
\colorlet{Highlight}{blue}

\tikzset{
modal/.style={>=stealth',shorten >=1pt,shorten <=1pt,auto,node distance=1.5cm,semithick},
world/.style={circle,draw,minimum size=1.0cm,fill=gray!15},
world_c/.style={circle,draw,minimum size=1.0cm,fill=green!60!black!40},
point/.style={circle,draw,inner sep=0.5mm,fill=black},
point_b/.style={circle,draw,double,inner sep=0.5mm,fill=white},
point_s/.style={circle,draw,inner sep=0.3mm,fill=black},
reflexive above/.style={->,loop,looseness=7,in=120,out=60},
reflexive below/.style={->,loop,looseness=7,in=240,out=300},
reflexive left/.style={->,loop,looseness=7,in=150,out=210},
reflexive right/.style={->,loop,looseness=7,in=30,out=330},
cross/.style={path picture={ 
  \draw[black]
(path picture bounding box.south east) -- (path picture bounding box.north west) (path picture bounding box.south west) -- (path picture bounding box.north east);
}},
cross_sum/.style={path picture={ 
  \draw[black]
(path picture bounding box.south) -- (path picture bounding box.north) (path picture bounding box.west) -- (path picture bounding box.east);
}}
}

\usepackage{amsthm}
\theoremstyle{plain}
\newtheorem{theorem}{Theorem}

\newtheorem{corollary}{Corollary}
\newtheorem{proposition}{Proposition}
\newtheorem{lemma}{Lemma}

\theoremstyle{definition}

\newtheorem{definition}{Definition}
\newtheorem{assumption}{Assumption}
\theoremstyle{remark}
\newtheorem{remark}{Remark}

\renewcommand\le{\leqslant}
\renewcommand\ge{\geqslant}
\newcommand\eps{\varepsilon}

\newcommand\R{\mathbb{R}}
\newcommand\Z{\mathbb{Z}}
\newcommand\N{\mathbb{N}}

\newcommand\Ss{\mathbb{S}}

\newcommand{\prenorm}[2]{\left\Vert #2 \right\Vert_{#1}}

\newcommand{\norm}[1]{\prenorm{}{#1}}
\newcommand\Sum{\sum\limits}
\newcommand\Int{\int\limits}
\newcommand\Prod{\prod\limits}

\DeclareMathOperator{\trace}{Tr}
\DeclareMathOperator{\real}{Re}

\DeclareMathOperator{\sgn}{sgn}
\DeclareMathOperator{\atan}{atan}

\graphicspath{{./images/}}

\begin{document}
%
\title{A Continuation Method for Large-Scale Modeling and Control: from ODEs to PDE, a Round Trip}
%
%
%

\author{Denis~Nikitin,
        Carlos Canudas-de-Wit
        and~Paolo~Frasca
\thanks{D. Nikitin, C. Canudas-de-Wit and P. Frasca are with Univ. Grenoble Alpes, CNRS, Inria, Grenoble INP, GIPSA-lab, 38000 Grenoble, France.}
}

\maketitle

\begin{abstract}
In this paper we present a continuation method which transforms spatially distributed ODE systems into continuous PDE. We show that this continuation can be performed both for linear and nonlinear systems, including multidimensional, space- and time-varying systems. When applied to a large-scale network, the continuation provides a PDE describing evolution of continuous state approximation that respects the spatial structure of the original ODE. 
Our method is illustrated by multiple examples including transport equations, Kuramoto equations and heat diffusion equations. As a main example, we perform the continuation of a Newtonian system of interacting particles and obtain the Euler equations for compressible fluids, thereby providing an original alternative solution to Hilbert's 6th problem. Finally, we leverage our derivation of the Euler equations to control multiagent systems, designing a nonlinear control algorithm for robot formation based on its continuous approximation.
\end{abstract}


%
\IEEEpeerreviewmaketitle

\section{Introduction}

\IEEEPARstart{M}{ost} of the systems we encounter in real life consist of such a large number of particles that the direct analysis of their interaction is impossible. In such cases, simplified models are used that aggregate the behavior of a set of particles and replace them with a continuous representation. 

The first model describing a system of moving and interacting particles was the Euler equation for liquids and gases which worked with a continuous field \cite{EULER}. This equation was one of the first \textit{Partial Differential Equations} (PDEs). Since then, many PDEs have been created that are models for describing different physical processes, such as the Maxwell equation for electromagnetic field \cite{MAXWELL}. 

With the advent of computers more and more attention has been paid to the discretization, that is, the process of transforming PDE into a system of ordinary differential equations in order to be able to numerically solve them \cite{DISCRETIZATION}. A correct discretization satisfies the condition that when the size of the resulting ODE system increases, its solution will converge to the original PDE. The choice of an apparent discretization method depends heavily on the PDE type such that the solution would preserve the properties of PDE \cite{CHOICE_DISCRETIZATION}.

Although discretization is a widely known and widely used method, the inverse problem of transforming an ODE system into PDE is more rarely studied. Our work focuses on this particular problem, with the aim of filling this gap and providing a counterpart to the discretization procedure. This can be useful since PDEs provide a much more compact way of describing the system, which in many cases is easier to analyze analytically than the corresponding ODE system.

The idea of replacing the system with its compact and simplified representation is widely used, especially for the ODE systems describing large-scale networks. Probably the most known approach of this type is a \textit{model reduction} technique which transforms a network into a smaller one while conserving the properties and the dynamics, see \cite{AOKI68, UMAR19}. Also the reduction towards the average state was studied in \cite{UMAR20, NIK20}.

Apart from various model reduction techniques large-scale networks are studied by \textit{mean field} methods in case of the \textit{all-to-all} interaction topology. In this situation the effect of the network on each node is the same, therefore it is enough to use an equation for a single agent together with parameters of a state of the whole network, see \cite{ACEBRON05} for a review with application to Kuramoto networks. 

The idea of mean field can be further extended to track not only a single agent's state, but the whole probability distribution over all agents' states in the network. This method is called \textit{population density} approach \cite{POPULATION00} and is mostly used to model large biological neural networks. In other fields of physics it is described as a \textit{probability density} of a physical system's state and is constructed by projection to probablistic space, see \cite{GRA82}.

Large-scale systems can be also simplified by studying the approximations to their probability densities, represented by \textit{moments}. E.g., in \cite{MOMENTS14, MOMENTS20} a moment-based approach was taken to control crowds dynamics. Different applications and issues of the method of moments are covered in \cite{MOMENTS_CL}. Also \textit{shape parameters} can be used to simplify the model and describe a shape of the solution as in \cite{MY_IFAC}.

Mean field and population density approaches are suitable in the case when the interaction topology between nodes is all-to-all. In other cases, the continuous representation of a system requires more sophisticated tools. A recently emerged theory of \textit{graphons} studies graph limits, i.e. structural properties that the graph possesses if the number of nodes tends to infinity while preserving interaction topology, see \cite{GRAPHONS1}. Using graphons it is possible to describe any dense graph as a linear operator in continuum space \cite{GRAPHONS2}. This method was further used to control large-scale linear networks \cite{GRAPHONS3} and to study sensitivity of epidemic networks \cite{GRAPHONS4}.

However, we are interested in systems that are spatially distributed and which have a position-dependent interaction, such as traffic in the city, power networks, robot formations, etc. By applying population density method or graphon theory to such system we would end up with a continuous model which looses the spatial structure of the problem.

Our idea is to replace the original spatially distributed ODE system by a continuous PDE whose state and space variables preserve the state and space variables of the original system. We develop a method for linear spatially invariant ODEs which transform them into PDEs with the help of finite differences. We name this method as a \textit{continuation}, since it is exactly opposite to the discretization procedure. Further we show how the continuation converges to the original system in sense of spectrum. Using computational graph formalism \cite{COMPUTATIONAL_GRAPH} we extend the method to nonlinear systems and further to space-dependent systems and systems with boundaries. 

The continuation method allows to recover a PDE which describes the same physical system as the original ODE network. Such a description can be very helpful both for analysis and control design purposes. Indeed, one can use an obtained PDE to design a continuous control which, being discretized back, results in a control law for the original ODE system: this design framework is illustrated in Fig.~\ref{fig_framework}.

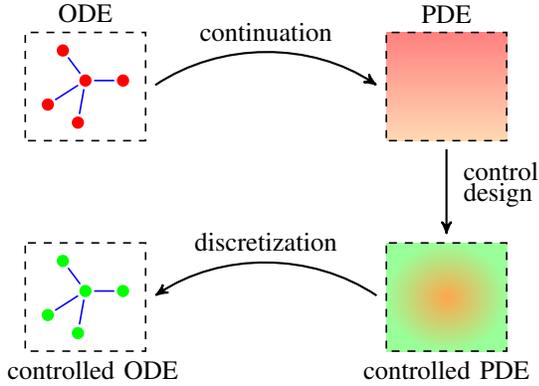
\begin{figure}
\begin{center}
\begin{tikzpicture}[modal,yscale=0.8]
	
	\node[point,red] (o11) at (0.2, 4) {};
	\node[point,red] (o12) at (0.5, 3.5) {};
	\node[point,red] (o13) at (1, 3.5) {};
	\node[point,red] (o14) at (0, 3.1) {};
	\node[point,red] (o15) at (0.4, 2.8) {};
	\path[-,blue] (o11) edge (o12);
	\path[-,blue] (o12) edge (o13);
	\path[-,blue] (o12) edge (o14);
	\path[-,blue] (o12) edge (o15);
	\draw[dashed] (-0.3,2.5) -- (1.3,2.5) -- (1.3,4.3) -- (-0.3,4.3) -- (-0.3,2.5);
	\node[above] at (0.5,4.3) {ODE};
	\path[thick,->]
(1.4,3.4) edge [out=40,in=140,looseness=1.0] (4.4,3.4);
	\node[above] at (2.9,4) {continuation};
	
	\draw[dashed,shading=axis,top color=red!50,bottom color=orange!30] (4.5,2.5) -- (6.1,2.5) -- (6.1,4.3) -- (4.5,4.3) -- (4.5,2.5);
      
    \node[above] at (5.3,4.3) {PDE};

  \draw[thick,->] (5.3,2.4) -- (5.3,0.9);
  \node[right] at (5.4,2) {control};
  \node[right] at (5.4,1.6) {design};
  
  	\draw[dashed,shading=radial,inner color=orange!70,outer color = green!40] (4.5,0.8) -- (6.1,0.8) -- (6.1,-1) -- (4.5,-1) -- (4.5,0.8);
  	
    
    \node[below] at (5.3,-1) {controlled PDE};

	\path[thick,->]
(4.4,-0.1) edge [out=140,in=40,looseness=1.0] (1.4,-0.1);
	\node[above] at (2.9,0.5) {discretization};

	\node[point,green] (o21) at (0.2, 0.5) {};
	\node[point,green] (o22) at (0.5, 0) {};
	\node[point,green] (o23) at (1, 0) {};
	\node[point,green] (o24) at (0, -0.4) {};
	\node[point,green] (o25) at (0.4, -0.7) {};
	\path[-,blue] (o21) edge (o22);
	\path[-,blue] (o22) edge (o23);
	\path[-,blue] (o22) edge (o24);
	\path[-,blue] (o22) edge (o25);
	\draw[dashed] (-0.3,0.8) -- (1.3,0.8) -- (1.3,-1) -- (-0.3,-1) -- (-0.3,0.8);

	 \node[below] at (0.6,-1) {controlled ODE};

\end{tikzpicture}
\end{center}
\caption{Proposed framework for control design based on the continuation method and a continuous representation of the system.} \label{fig_framework}
\end{figure}

Moving to the examples of the application of the continuation method, we tackle the Hilbert's 6th problem, questioning how one can rigorously transform a system of interacting particles into the Euler PDE. We provide our treatment of this problem using continuation, deriving the Euler PDE from a system of Newton's laws for the case of long-range interaction forces. Further, we show how the method can be applied to the multiagent control, providing a simple control algorithm to stabilize a robotic formation along the desired trajectory, performing a maneuver of passing through a window. The control is derived on a level of a PDE representation and then it is discretized to be implemented on every agent in accordance with the scheme in Fig.~\ref{fig_framework}.

We start in Section II by defining a continuation for linear ODEs, discussing questions of accuracy, convergence and choice of the particular model. Section III continues to nonlinear models, utilizing the computational graph formalism. In Section IV the method is extended to much broader class of systems, including multidimensional or space- and time-varying systems and also discussing boundary conditions. Section V is devoted to a derivation of the Euler PDE from a system of interacting particles. Finally, Section VI applies the method to control a robotic formation.

\section{Method for linear systems} \label{SEC_method_linear}

The simplest class of systems for which the transformation of ODE into PDE can be performed is given by linear ODE systems corresponding to the dynamics of states of nodes, which are aligned on the 1D line in space and depend only on some fixed set of their neighbours.
\begin{figure}
\begin{center}
\begin{tikzpicture}[modal]
	
	\node[circle,draw,minimum size=1.1cm,fill=green!60!black!40] (x0) {$\rho_i$};
	\node[circle,draw,minimum size=1.1cm,fill=gray!15] (x1) [right=0.3cm of x0] {$\rho_{i+s_2}$};
	\node[circle,draw,minimum size=1.1cm,fill=gray!15] (xm1) [left=0.3cm of x0] {$\rho_{i+s_1}$};
	\node[point_s] (p1) [right=0.2cm of x1] {};
	\node[point_s] (p2) [right=0.4cm of x1] {};
	\node[point_s] (p3) [right=0.6cm of x1] {};
	\node[point_s] (pm1) [left=0.2cm of xm1] {};
	\node[point_s] (pm2) [left=0.4cm of xm1] {};
	\node[point_s] (pm3) [left=0.6cm of xm1] {};
	
	\node[circle,draw,minimum size=1.1cm,fill=gray!15] (x2) [right=0.2cm of p3] {$\rho_{i+s_j}$};
	\node[point_s] (pp1) [right=0.2cm of x2] {};
	\node[point_s] (pp2) [right=0.4cm of x2] {};
	\node[point_s] (ppm3) [right=0.6cm of x2] {};
	
	\path[<-] (x0) edge[in=120,out=60] node[above right]{$a_2$} (x1);
	\path[<-] (x0) edge[in=60,out=120] node[above]{$a_1$} (xm1);
	\path[<-] (x0) edge[in=120,out=75] node[above]{$a_j$} (x2);
	
	\path[->] ($(x0)+(-2.5,-1.1)$) edge node[pos=1,right]{$x$} ($(x0)+(4.5,-1.1)$);
	
	\path[-] ($(x0)+(-1.4,-1.2)$) edge node[below]{$x_{i-1}$} ($(x0)+(-1.4,-1.0)$);
	\path[-] ($(x0)+(0,-1.2)$) edge node[below]{$x_i$} ($(x0)+(0,-1.0)$);
	\path[-] ($(x0)+(1.4,-1.2)$) edge node[below]{$x_{i+1}$} ($(x0)+(1.4,-1.0)$);
	\path[-] ($(x0)+(3.45,-1.2)$) edge node[below]{$x_{i+s_j}$} ($(x0)+(3.45,-1.0)$);
	
	\node at ($(x0)+(-0.7,-0.85)$) {$\Delta x$};
	\node at ($(x0)+(0.7,-0.85)$) {$\Delta x$};
	
\end{tikzpicture}
\end{center}
\vskip -10pt
\caption{System of nodes aligned in 1D line with dynamics given by \eqref{linear_ODE} with $s_1 = -1$ and $s_2 = 1$.} \label{FIG_linear}
\end{figure}
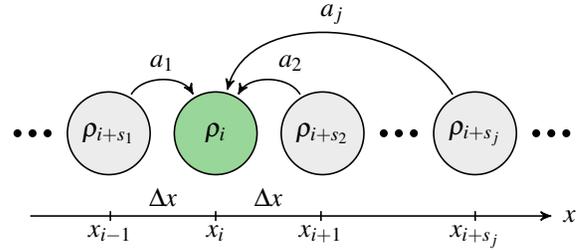
Let the node $i$ have a state $\rho_i \in \R$ and a geographical position $x_i \in \R$ such that for every $i$ the distance between two consecutive nodes in space is constant, $x_{i+1} - x_i = \Delta x$ (the assumption of $\Delta x$ being constant will be relaxed later on).
Then the systems of our interest take the form
\begin{equation} \label{linear_ODE}
\dot\rho_i = \Sum_{j=1}^N a_j \rho_{i+s_j},
\end{equation}
that is $\dot\rho_i$ linearly depends only on $N$ neighbouring nodes shifted by $s_j \in \Z$ for $j \in \{1..N\}$, and $a_j \in \R$ are the system gains, see Fig.~\ref{FIG_linear}. This type of systems belongs to the class of linear \textit{spatially invariant systems} \cite{BASSAM}, which is a natural class for distributed control.

\subsection{Motivating Example}

We start by considering the most simple ODE system of class \eqref{linear_ODE} which has spatial dependence:
\begin{equation} \label{Transport_ODE}
\dot \rho_i = \frac{1}{\Delta x} \left(\rho_{i+1} - \rho_i \right).
\end{equation}
Comparing with \eqref{linear_ODE}, here $N = 2$, $a_1 = 1/\Delta x$, $a_2 = -1/\Delta x$, $s_1 = 1$ and $s_2 = 0$. This equation describes a transport of some quantity along the line, and is usually referred as a Transport ODE. Equation \eqref{Transport_ODE} is sometimes studied on its own, but often it comes as a result of a discretization process applied to another equation,
\begin{equation} \label{Transport_PDE}
\frac{\partial \rho}{\partial t} = \frac{\partial \rho}{\partial x}.
\end{equation}
This equation belongs to a class of PDEs, which is usually thought to be more difficult class of equations to study than ODEs. However, equation \eqref{Transport_PDE} describes a perfect transport of information with finite propagation speed along the line, which can be studied much more easily in PDE form than in ODE, as it perfectly conserves the form of a solution, performing only a shift along the line as time increases. We will refer to this equation as a Transport PDE. 

Equation \eqref{Transport_ODE} can be obtained from \eqref{Transport_PDE} by the discretization process, which has been a well-established mathematical tool for centuries. Nevertheless, up to now there was no strict procedure describing a general process which could render equation \eqref{Transport_PDE} from \eqref{Transport_ODE}. In the next subsections we explore more how the general discretization procedure is defined for linear systems and how it should be inverted to obtain a continuation process.

\subsection{Discretization}

The discretization of PDEs is usually performed by a finite difference method, where the partial derivatives are approximated by finite differences. For example, in the case of Transport ODE, 
$$
\frac{\partial \rho}{\partial x} \approx \frac{1}{\Delta x} \left(\rho_{i+1} - \rho_i \right).
$$
This approximation is valid when $\Delta x$ is small. Indeed, assuming that the solution to PDE is given by a smooth function $\rho(x)$ and using Taylor series, we can write
\begin{equation} \label{Taylor_i_p_1}
\rho_{i+1} = \rho(x_{i+1}) = \rho(x_i) + \frac{\partial \rho}{\partial x} \Delta x + \frac{\partial^2 \rho}{\partial x^2} \frac{\Delta x^2}{2} + ...,
\end{equation} 
where all partial derivatives are calculated in $x_i$. Thus, subtracting $\rho_i$ and dividing by $\Delta x$, we get
\begin{equation} \label{Discretization_FD1}
\frac{\partial \rho}{\partial x} = \left[\frac{1}{\Delta x} \left(\rho_{i+1} - \rho_i \right)\right] - \frac{\partial^2 \rho}{\partial x^2} \frac{\Delta x}{2} - ...,
\end{equation}
which means that the residual belongs to the class $O(\Delta x)$ of all functions which go to zero at least as fast as $\Delta x$. Thus, taking $\Delta x$ sufficiently small, one can ensure the arbitrary accuracy of the approximation, provided all the partial derivatives are bounded. 

Accuracy can be further increased by taking different points where the function is sampled, called stencil points. For example, writing
\begin{equation} \label{Taylor_i_m_1}
\rho_{i-1} = \rho(x_{i-1}) = \rho(x_i) - \frac{\partial \rho}{\partial x} \Delta x + \frac{\partial^2 \rho}{\partial x^2} \frac{\Delta x^2}{2} - ...,
\end{equation} 
subtracting \eqref{Taylor_i_m_1} from \eqref{Taylor_i_p_1} and dividing by $2\Delta x$, we get
\begin{equation} \label{Discretization_FD2}
\frac{\partial \rho}{\partial x} = \left[\frac{1}{2\Delta x} \left(\rho_{i+1} - \rho_{i-1} \right)\right] - \frac{\partial^3 \rho}{\partial x^3} \frac{\Delta x^2}{6} + ....
\end{equation}
Thus, using stencil points $\{i-1, i+1\}$ to approximate the first-order derivative in the point $i$ the obtained residual belongs to the class $O(\Delta x^2)$, which means that this discretization of the Transport PDE is accurate to the second order.

In general, if one wants to approximate the derivative of order $m$ in point $i$ using $N$ stencil points $\{i+s_1, i+s_2, ..., i+s_N\}$ with $m < N$ in form
\begin{equation} \label{Discretization}
\frac{\partial^m \rho}{\partial x^m} \approx \Sum_{j=1}^N \hat a_j \rho_{i+s_j}
\end{equation}
where coefficients $\hat a_j$ are unknown, one can define $S_{N,N} \in \R^{N\times N}$, $\hat a \in \R^N$ and $c \in \R^N$ by
\begin{equation*}
S_{N,N} = \begin{pmatrix}
1 & \cdots & 1 \\
s_1 & \cdots & s_N \\
\vdots & \ddots & \vdots \\
s_1^{N-1} & \cdots & s_N^{N-1}
\end{pmatrix}, \;
\hat a = \begin{pmatrix}
\hat a_1 \\ \hat a_2 \\ \vdots \\ \hat a_N
\end{pmatrix}, \; c = \frac{m!}{\Delta x^m} \begin{pmatrix}
0 \\ \vdots \\ 1 \\ \vdots \\ 0 \end{pmatrix},
\end{equation*}
where $c$ has $1$ on the position $m+1$, and solve a linear system
\begin{equation} \label{Discretization_LS}
\hat a = S_{N,N}^{-1} \: c.
\end{equation}
The system \eqref{Discretization_LS} can be trivially obtained by writing Taylor series for all points $\rho_{i+s_1}...\rho_{i+s_N}$ and summing them in a linear combination as in \eqref{Discretization}. The obtained order of accuracy is at least $O(\Delta x^{(N-m)})$, and sometimes can be higher if some of the higher derivatives are also eliminated (as in case of \eqref{Discretization_FD2}).

\subsection{Continuation} \label{SEC_linear_continuation}

Essentially the same process can be applied to the equation \eqref{linear_ODE} to get the PDE version. For every term in the summation in \eqref{linear_ODE} we can write
\begin{equation} \label{Taylor_i_p_sj}
\rho_{i+s_j} = \rho(x_{i+s_j}) = \rho(x_i) + \frac{\partial \rho}{\partial x} \Delta x s_j + \frac{\partial^2 \rho}{\partial x^2} \frac{\Delta x^2 s_j^2}{2} + ...
\end{equation}
Thus, assume we state the problem of finding the PDE approximation of \eqref{linear_ODE} in form
\begin{equation} \label{linear_combination_continuation}
\Sum_{j=1}^N a_j \rho_{i+s_j} \approx \Sum_{k=0}^d c_k \frac{\Delta x^k}{k!} \frac{\partial^k \rho}{\partial x^k},
\end{equation}
where $d$ is the highest order of derivative (\textit{order of continuation}) we want to use. Note that zero is also included in the right summation, since the function itself can be used in the resulting PDE. Then, introducing $S_{d+1,N} \in \R^{(d+1)\times N}$, $a \in \R^N$ and $c \in \R^{d+1}$ by
\begin{equation} \label{Continuation_matrices}
S_{d+1,N} = \begin{pmatrix}
1 & \cdots & 1 \\
s_1 & \cdots & s_N \\
\vdots  & \ddots & \vdots \\
s_1^d & \cdots & s_N^d
\end{pmatrix}, \;\;\;
a = \begin{pmatrix}
a_1 \\ a_2 \\ \vdots \\ a_N
\end{pmatrix}, \;\;\; c = \begin{pmatrix}
c_0 \\ c_1 \\ \vdots \\ c_d \end{pmatrix},
\end{equation}
the vector of unknown coefficients $c$ can be found by direct multiplication,
\begin{equation} \label{Continuation_LS}
c = S_{d+1,N} \; a, \quad \text{or} \quad c_k = \Sum_{j=1}^N a_j s_j^k \quad \forall k \in \{0,..,d\}.
\end{equation}
Once \eqref{Continuation_LS} is solved, we write the PDE approximation to \eqref{linear_ODE}:
\begin{equation} \label{linear_PDE}
\frac{\partial \rho}{\partial t} = \Sum_{k=0}^d c_k \frac{\Delta x^k}{k!} \frac{\partial^k \rho}{\partial x^k}.
\end{equation}
As an example, applying \eqref{Continuation_LS} to the Transport ODE \eqref{Transport_ODE} renders coefficients $c_0 = 0$ and $c_k = 1/\Delta x$ for all $k > 0$, and choosing $d = 1$ we obtain the Transport PDE \eqref{Transport_PDE}.

\subsection{Accuracy of continuation}

Procedures \eqref{Discretization_LS} and \eqref{Continuation_LS} look very similar from the algebraic point of view, however they are qualitatively different in the way how the problem is formulated and how we should interpret their results.

The discretization procedure tries to find the best approximation to a \textit{continuous and smooth} function $\rho$ and its derivatives. What is most important, the discretization step $\Delta x$ is usually an adjustable parameter which can be set by a system engineer \textit{arbitrarily small} to satisfy the desired performance. Thus the notion of accuracy of a discretization is used to describe how fast the solution of the discretized equation tends to the solution of the original equation when $\Delta x$ tends to zero. In some sense this means quality of the discretization, since the higher order of accuracy means that the engineer can take a larger $\Delta x$ to achieve the same error and thus use the smaller number of states in the discretized system.

Instead, when the original system is given by the ODE, the nodes have fixed locations, thus $\Delta x$ is a \textit{true constant} representing properties of an underlying physical system and it cannot be changed by an engineer. In turn this means that the accuracy defined as a class $O(\Delta x^{(N-m)})$ cannot measure quality of the approximation as $\Delta x$ does not behave as an arbitrarily small value. Moreover, in the ODE case the system state $\rho_i$ is known only on a given set of points $i$, thus in general the approximation $\rho(x)$ can be \textit{non-smooth or discontinuous}. Even if we assume the smoothness at the initial moment of time, the dynamics can render its derivatives unbounded. As a result the series in \eqref{Taylor_i_p_sj} can be non-convergent.

We know however that the systems \eqref{linear_ODE} and \eqref{linear_PDE} are connected by finite difference methods \eqref{Discretization_LS} and \eqref{Continuation_LS}.  We will use this fact as a definition of a PDE approximation to an ODE and of an accuracy of this approximation.

\begin{definition}
Discretization of PDE to ODE is called \textit{valid} if it is performed according to the finite difference method \eqref{Discretization_LS}.
\end{definition}

\begin{definition}
Continuation of ODE to PDE is called \textit{valid} if there exists a valid discretization of the obtained PDE to the original ODE.
\end{definition}

Definitions 1 and 2 say that if we use continuation on some ODE and then perform discretization at the same stencil points, we should arrive at the same ODE. This procedure sets a constraint on the minimum order of the PDE:
\begin{theorem}
The PDE \eqref{linear_PDE} that is obtained from the ODE \eqref{linear_ODE} with $N$ stencil points is valid if and only if
\begin{equation} \label{linear_constraint}
d + 1 \ge N.
\end{equation}
\end{theorem}

\noindent\textit{Proof of necessity: }
Indeed, assume $d + 1 < N$. This means that the PDE approximation \eqref{linear_PDE} has the highest derivative at most of the order $d$. Thus we can augment the vector of coefficients $c$ by $N - d - 1$ zeros corresponding to higher-order derivatives, which is equivalent to the augmentation of $S_{d+1,N}$ with $N - d - 1$ zero rows since $c$ was defined by \eqref{Continuation_LS}. Now, applying the discretization process \eqref{Discretization_LS} to $c$ we should arrive at the same vector $a$ of the parameters of the ODE system. Since this should be true for any $a$, we substitute $S_{N,N}^{-1}$ and augmented $S_{d+1,N}$ and obtain a condition
$$
\footnotesize{
\begin{pmatrix}
1 & 1 & \cdots & 1 \\
s_1 & s_2 & \cdots & s_N \\[-1mm]
\vdots & \vdots & \ddots & \vdots \\
s_1^{N-1} & s_2^{N-1} & \cdots & s_N^{N-1}
\end{pmatrix}^{-1} \begin{pmatrix}
1 & 1 & \cdots & 1 \\
s_1 & s_2 & \cdots & s_N \\[-2mm]
\vdots & \vdots & \ddots & \vdots \\[-0.5mm]
s_1^{d} & s_2^{d} & \cdots & s_N^{d} \\
0 & 0 & \cdots & 0 \\[-2mm]
\vdots & \vdots & \ddots & \vdots
\end{pmatrix} = I,}
$$
which is impossible to satisfy since the second matrix is singular. Therefore there is no valid discretization process for the PDE obtained by continuation with order $d$ such that $d + 1 < N$, which by definition means that such continuation is not valid.

\noindent\textit{Proof of sufficiency: }
Case $d + 1 = N$ is trivial, since the equations \eqref{Discretization_LS} and \eqref{Continuation_LS} are equivalent in this situation. This obviously provides a validity of the continuation procedure.

Now assume $d + 1 > N$. Then the obtained vector of coefficients $c$ is of higher dimension than $a$. The discretization \eqref{Discretization_LS} cannot be applied directly, since there is not enough stencil points to express the finite differences for the derivatives of order higher than $N - 1$. However we can increase the set of stencil points. Let us choose additional $d + 1 - N$ stencil points $\bar s_{N + 1}, \bar s_{N + 2}, ..., \bar s_{d + 1}$. Applying continuation \eqref{Continuation_LS} to the original ODE \eqref{linear_ODE} and then \eqref{Discretization_LS} to the obtained PDE using the augmented set of stencil points we get a new ODE gains $\bar a$ which are expressed as
$$
\footnotesize{
\bar a = \begin{pmatrix}
1 & \cdots & 1 & 1 & \cdots & 1 \\
s_1 & \cdots & s_N & \bar s_{N+1} & \cdots & \bar s_{d + 1} \\
\vdots & \ddots & \vdots & \vdots & \ddots & \vdots \\
s_1^{d} & \cdots & s_N^{d} & \bar s_{N+1}^d & \cdots & \bar s_{d + 1}^d
\end{pmatrix}^{-1}\begin{pmatrix} 
1 & \cdots & 1 \\
s_1 & \cdots & s_N \\
\vdots & \ddots & \vdots \\
s_1^d  & \cdots & s_N^d
\end{pmatrix} a.}
$$
It is now clear that the first $N$ elements of $\bar a$ are exactly $a$ and the rest is zero, irrespective of the chosen additional points $\bar s_j$. Thus the artificially introduced stencil points do not appear in the discretized PDE, rendering the same ODE as the original one. \qed

The latter part of the proof of Theorem 1 means that the PDE obtained by the process \eqref{Continuation_LS} with $d + 1 > N$ has more information that one with $d + 1 = N$, since it provides exact Taylor approximations not only on the given set of points, but in the additional $d + 1 - N$ points which can be chosen arbitrary. This property can be used to define the order of accuracy as $d + 1 - N$.

\begin{definition}
Order of accuracy of a valid continuation process of ODE to PDE is defined as the number of additional points in which the corresponding discretization process can be made exact simultaneously, i.e. $d + 1 - N$.
\end{definition}
For example, a continuation
$$
\rho_{i+1} - \rho_i \quad \rightarrow \quad \Delta x \frac{\partial \rho}{\partial x}
$$
is of order of accuracy 0, since trying to discretize the PDE on any larger set of stencil points except from $\{i, i+1\}$ will give different ODE. At the same time a continuation
$$
\rho_{i+1} - \rho_{i-1} \quad \rightarrow \quad 2\Delta x \frac{\partial \rho}{\partial x} + \left(0 \cdot \frac{\partial^2 \rho}{\partial x^2} \right)
$$
has order of accuracy 1 because the second derivative vanishes (thus $d = 2$), and it is possible to discretize the PDE on a set of stencil points of size 3 (with one additional point), for example $\{i-1, i, i+1\}$. From now on we will consider only valid continuations, and most often in practice we will use continuations of order of accuracy $0$.

\subsection{Convergence of continuation}

It is clear that the higher order of continuation is taken, the better the original ODE operator \eqref{linear_ODE} is approximated by the PDE \eqref{linear_PDE}. It is possible to study the convergence properties by shifting the problem to the frequency domain using the Fourier transform.

Let us define the function $a(x)$ as 
$$
a(x) = \Sum_{j=1}^N a_j \delta\left(x + s_j \Delta x \right),
$$
where $\delta(x)$ is the Dirac delta function. Therefore, for any integrable function $\rho(x)$ the equation \eqref{linear_ODE} is equivalent to the following system with convolution
\begin{equation} \label{linear_ODE_conv}
\dot\rho(x) = (a \star \rho)(x).
\end{equation}
Use now the Fourier transform, defined as 
\begin{equation}
\hat f(\omega) = \int_{-\infty}^{\infty} f(x) e^{-i x \omega} dx
\end{equation}
for any integrable function $f(x)$ and for any frequency $\omega \in \R$. It is known that the Fourier image of a convolution is a multiplication. Therefore the system \eqref{linear_ODE_conv} is just
\begin{equation} \label{spectrum_ODE}
\dot{\hat\rho}(\omega) = \hat a(\omega) \hat\rho(\omega), \qquad \hat a(\omega) = \Sum_{j=1}^N a_j e^{i s_j \Delta x \omega},
\end{equation}
where $\hat a(\omega)$ was found by direct calculation of Fourier transform. \eqref{spectrum_ODE} immediately shows that the spectrum of the original ODE system is parametrized by $\hat a(\omega)$. In fact, this result is well-known, since the system \eqref{linear_ODE} on an infinite line belongs to the class of Laurent systems, whose spectrum is known to be \eqref{spectrum_ODE}, see \cite{LAURENT}.

Now let us calculate the spectrum of the continualized system \eqref{linear_PDE}. By another property of the Fourier transform, if the function $\rho(x)$ is sufficiently smooth and its derivatives are integrable, we can recover their Fourier images by
$$
\widehat {\left( \frac{\partial^k \rho}{\partial x^k}\right)} \left(\omega\right) = \left(i\omega\right)^k \hat\rho(\omega).
$$
Therefore \eqref{linear_PDE} is read in frequency domain as
\begin{equation} \label{spectrum_PDE}
\dot{\hat\rho}(\omega) = \hat c(\omega) \hat\rho(\omega), \qquad \hat c(\omega) = \Sum_{k=0}^d c_k \frac{\Delta x^k}{k!} \left(i\omega\right)^k.
\end{equation}
Substituting \eqref{Continuation_LS}, we can rewrite \eqref{spectrum_PDE} as
\begin{equation} \label{spectrum_PDE_exp}
\hat c(\omega) = \Sum_{j=1}^N a_j \Sum_{k=0}^d \frac{(i s_j \Delta x \omega)^k}{k!}.
\end{equation}

Now, comparing \eqref{spectrum_PDE_exp} with \eqref{spectrum_ODE}, it is clear that \eqref{spectrum_PDE_exp} uses the first $d + 1$ terms of the Taylor expansion of the exponential function in \eqref{spectrum_ODE}. In fact, since the exponential function is analytic on the whole complex plane, we have just proven the following result:
\begin{theorem}
The spectrum of the PDE \eqref{linear_PDE} converges to the spectrum of the original ODE \eqref{linear_ODE} pointwise as $d \to \infty$.
\end{theorem}

However, the convergence of spectrums is not uniform. Moreover, the spectrum \eqref{spectrum_ODE} is an image of the unit circle and thus is a compact set, while the spectrum \eqref{spectrum_PDE_exp} for any $d$ is a polynomial and thus unbounded. It is still possible to state several corollaries of Theorem 2:
\begin{corollary}
For any bounded subset of frequencies the spectrum of the PDE \eqref{spectrum_PDE_exp} converges to the spectrum of the ODE \eqref{spectrum_ODE} uniformly.
\end{corollary}
\begin{proof}
The result follows directly from Theorem 2 and the fact that the Taylor expansion of any analytic function converges to the function itself uniformly on any bounded domain.
\end{proof}
\begin{figure}
\begin{center}
\includegraphics[width=\linewidth]{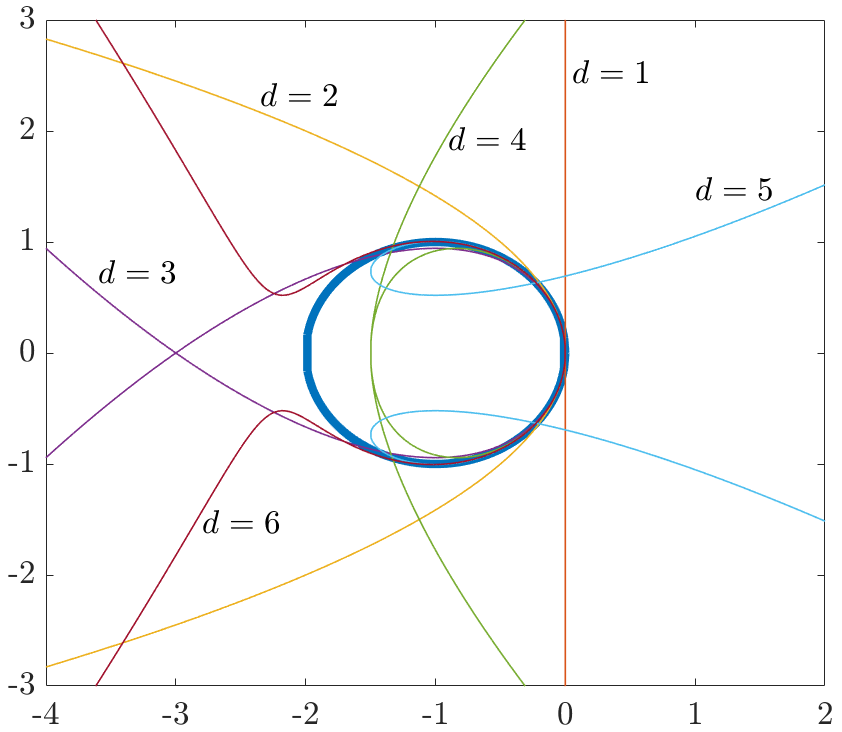}
\caption{Spectrum for the Transport ODE \eqref{Transport_ODE_2} $e^{i\omega}-1$ (blue circle) together with spectrums of the continuations up to the order 6, including \eqref{Transport_PDE_2}. As $d$ increases, spectrums converge to the blue circle, however for some orders (such as 4 or 5) they can become unstable.} \label{FIG_transport_PDE_spectrum}
\end{center}
\end{figure}
\begin{corollary}
If the original ODE \eqref{linear_ODE} is unstable, there exists $D \ge 0$ such that for all $d \ge D$ the continualized system \eqref{linear_PDE} will also be unstable.
\end{corollary}
\begin{proof}
Since the original system is unstable, there exists $\omega_0$ such that $\real \hat a(\omega_0) > 0$. Now, by the definition of a limit there exists $D \ge 0$ such that for all $d \ge D$  $\real \hat c(\omega_0) > 0$.
\end{proof}
Unfortunately, the stability of the original system cannot be always recovered by taking large enough $d$. The reason for this is in the unboundedness of the polynomial spectrum \eqref{spectrum_PDE_exp}. Even if the original system is stable, for some $k$ it is possible to introduce artificial instability on high frequencies. There are however some guidelines which can be used to properly choose the order of continuation $d$:
\begin{corollary}
PDE \eqref{linear_PDE} with an odd order of continuation $d$ has the same stability properties as a PDE with the order of continuation $d-1$.
\end{corollary}
\begin{proof}
All odd terms in the spectrum \eqref{spectrum_PDE_exp} are purely imaginary and thus have no impact on the stability.
\end{proof}
\begin{corollary}
Artificial instability is introduced when the last even term in the PDE \eqref{linear_PDE} has $c_k > 0$ if $k = 4m$ or $c_k < 0$ if $k = 4m + 2$ for some $m \in \Z^+$.
\end{corollary}
\begin{proof}
Artificial instability comes if the term of the polynomial \eqref{spectrum_PDE_exp} with the highest even power is positive, which leads to a positive real part of a spectrum on high frequencies. Positivity of the highest even term is exactly equivalent to the statement of the corollary since $i^{4m} = 1$ and $i^{4m+2} = -1$ for any $m \in \Z^+$.
\end{proof}

We will demonstrate the convergence of spectrums on the Transport ODE
\begin{equation} \label{Transport_ODE_2}
\dot\rho_i = \rho_{i+1} - \rho_i.
\end{equation}
Assuming $\Delta x = 1$, the continuation of \eqref{Transport_ODE_2} is:
\begin{equation} \label{Transport_PDE_2}
\frac{\partial \rho}{\partial t} = \Sum_{k=1}^d \frac{1}{k!}  \frac{\partial^k \rho}{\partial x^k},
\end{equation}
Spectrum of \eqref{Transport_ODE_2} by formula \eqref{spectrum_ODE} is given by $e^{i\omega} - 1$, which is depicted as a blue circle in Fig.~\ref{FIG_transport_PDE_spectrum} together with the spectrums of the continuations up to the order $d=6$. It is clear that as the order increases, the approximations become better. 

The original Transport ODE is stable. Moreover, it has an intrinsic diffusion in it, which can be captured by the continuation of the second order. However, the continuation of order 4 is unstable. It happens because of an artificial instability as described in Corollary 4, since $c_4 = 1 > 0$. In general all stable continuations of the Transport ODE are given by the orders $\{1,2,3, ..., 4m + 2, 4m + 3, ... \}$ for all $m \in \Z^+$.

\section{Method for nonlinear systems} \label{SEC_nonlinear}

Finite differences give us a complete tool for linear systems, but for nonlinear systems they should be applied in composition with nonlinearities. Using an additional concept of computational graph it is possible to elaborate the case of general nonlinear ODE systems.

As in the previous case we assume without loss of generality that the nodes are equally spaced along the 1D line, a node $i$ having a state $\rho_i$ and a position $x_i$. Then the general nonlinear ODE with space dependence takes form of
\begin{equation} \label{nonlinear_ODE}
\dot \rho_i = F(\rho_{i+s_1}, \rho_{i+s_2}, ..., \rho_{i+s_N}).
\end{equation}
We further assume that the function $F$ is continuous.

\subsection{Computational graph}

In 1957 Kholmogorov \cite{KHOLMOGOROV} showed that every multidimensional continuous function can be written as a composition of functions of one variable and additions. This work laid the basis for the neural networks function approximation, which is now a major branch of modern machine learning.

Here we will use this idea and assume that the function $F$ is given in the form of computational graph (see \cite{COMPUTATIONAL_GRAPH} for example). This is a directed acyclic graph, every node of which represents a one-dimensional function, applied to a weighted sum of inputs coming to this node. We assume that the leaves of this graph are the states of the system $\rho_{i+s_j}$ and the root node computes the resulting value of $F$. 

As an example of the computation graph we will consider a system
\begin{equation} \label{Kuramoto_ex}
\dot \rho_i = \sin(\rho_{i+1} - \rho_i) - \sin(\rho_i - \rho_{i-1})
\end{equation}
which is a system of Kuramoto oscillators coupled on a ring. The computational graph for \eqref{Kuramoto_ex} is presented in Fig.~\ref{FIG_graph_sin}.
\begin{figure}
\begin{center}
\begin{tikzpicture}[modal]
	
	\node[world_c] (res) {\text{result}};
	
	\node[world] (s1) [below left=1.0cm of res, xshift=-0.6cm] {\text{sin}};
	\node[world] (s2) [below right=1.0cm of res, xshift=0.6cm] {\text{sin}};
	
	\node[world] (p1) [below left=1.0cm of s1, xshift=0.3cm] {$\rho_{i-1}$};
	\node[world] (p2) [below right=1.0cm of s1, xshift=-0.3cm] {$\rho_i$};
	\node[world] (p3) [below left=1.0cm of s2, xshift=0.3cm] {$\rho_i$};
	\node[world] (p4) [below right=1.0cm of s2, xshift=-0.3cm] {$\rho_{i+1}$};
	
	\path[->] (s1) edge node[above left]{-1} (res);
	\path[->] (s2) edge node[above right]{1} (res);
	
	\path[->] (p1) edge node[above left]{-1} (s1);
	\path[->] (p2) edge node[above right]{1} (s1);
	\path[->] (p3) edge node[above left]{-1} (s2);
	\path[->] (p4) edge node[above right]{1} (s2);
	
	\draw[rounded corners, dashed, blue] (-4, -0.8) rectangle (-0.1, -3.7) {};
	\draw[rounded corners, dashed, blue] (4, -0.8) rectangle (0.1, -3.7) {};
	
	\draw[rounded corners, dashed, red] (-3.8, -2.3) rectangle (-2.6, -3.5) {};
	\draw[rounded corners, dashed, red] (-1.55, -2.3) rectangle (-0.35, -3.5) {};
	\draw[rounded corners, dashed, orange] (1.55, -2.3) rectangle (0.35, -3.5) {};
	\draw[rounded corners, dashed, orange] (3.8, -2.3) rectangle (2.6, -3.5) {};
	
	\node (u1) [left=0.4cm of s1, yshift=0.3cm] {$i-\frac{1}{2}$};
	\node (u1) [right=0.4cm of s2, yshift=0.3cm] {$i+\frac{1}{2}$};
\end{tikzpicture}
\end{center}
\caption{Computation graph for the system \eqref{Kuramoto_ex}. Similar subgraphs are outlined by dashed rectangles of the same color. Possible choices of sinus subgraph' positions are written in the corners of blue rectangles.} \label{FIG_graph_sin}
\end{figure}
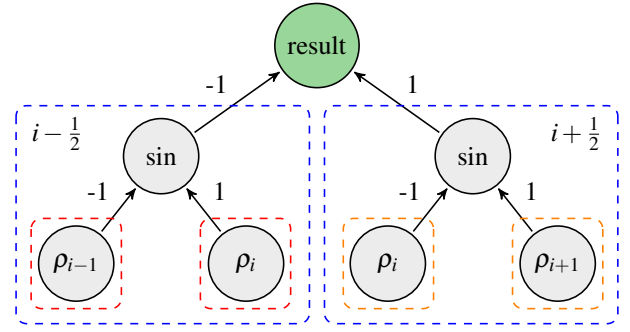

\subsection{Similar subgraphs and their positions}

Now let us introduce an original notion of \textit{similar subgraphs}. Subgraph is a computational graph which computes subexpression of the original computational graph. Every node in a computational graph serves as the root of a subgraph computing expression defined in this node. The leaf nodes are also the subgraphs ''computing'' themselves.

\begin{definition} 
We call two subgraphs \textit{similar} if
\begin{enumerate}
\item they serve as an input to the same node,
\item they differ only in the positions of the leaf nodes, and this difference can be represented by a single shift.
\end{enumerate}
\end{definition}
This is an equivalence relation, therefore we can speak about equivalence classes which we call sets of similar subgraphs.

For example, in Fig.~\ref{FIG_graph_sin} there are three sets of similar subgraphs:
\begin{enumerate}
\item $\rho_{i-1}$ and $\rho_i$ for the left sinus node,
\item $\rho_i$ and $\rho_{i+1}$ for the right sinus node,
\item $\sin(\rho_i - \rho_{i-1})$ and $\sin(\rho_{i+1} - \rho_i)$ for the root node, because they differ by a single shift which equals 1.
\end{enumerate}

The last thing which should be defined is a \textit{position} of a subgraph:
\begin{definition}
Position of a subgraph is defined as a coordinate in space where the expression of this subgraph is calculated.
\end{definition}

The leaf nodes by definition are the states of the system, thus they are calculated at some positions on the line. For example the leaf node $\rho_{i+1}$ in Fig.~\ref{FIG_graph_sin} is defined in the point $x_{i+1}$, thus we will say that its position is $i+1$. 

The root node by definition has a position $i$, since it is exactly the position of the left-hand side term in \eqref{nonlinear_ODE}. We will define the positions of other subgraphs as the average of their leaves' positions. Note that in general there is some freedom in the definition of the subgraphs' positions, with the only constraint that similar subgraphs should differ by a single shift, but we will omit this for simplicity.

Since the position of a subgraph represents a position on the line, it is natural to have non-integer position values, although the leaf nodes and the root have only integer positions. As an example, defining the position as an average, in Fig.~\ref{FIG_graph_sin} the node $\sin(\rho_{i+1} - \rho_i)$ has its position $i + 1/2$.

\subsection{Continuation to a nonlinear PDE}

When system \eqref{nonlinear_ODE} is expressed in a form of computational graph with similar subgraphs being found and their positions being defined, one can perform a continuation procedure described in section \ref{SEC_linear_continuation} to obtain a PDE.

Continuation should be performed recursively, starting from the leaves. Each set of similar subgraphs by definition is used in their common ancestor node as a linear combination of equivalent elements shifted by some distance. Continuation of this linear combination by \eqref{linear_combination_continuation} replaces the set of similar subgraphs by a weighted sum of partial derivatives of subexpressions, calculated at the position of the ancestor node.

Let $\Delta x = x_{i+1} - x_i$ be a distance between any two neighbouring nodes. Elaborating example \eqref{Kuramoto_ex} and using order of accuracy 0, we perform the continuation in three steps: 
\begin{enumerate}
\item $\sin(\rho_{i+1} - \rho_i) \quad \rightarrow \quad \sin\left(\Delta x\dfrac{\partial \rho}{\partial x}(x_{i + 1/2})\right)$,
\item $\sin(\rho_i - \rho_{i-1}) \quad \rightarrow \quad \sin\left(\Delta x\dfrac{\partial \rho}{\partial x}(x_{i - 1/2})\right)$,
\item $\sin_{i+1/2} - \sin_{i-1/2} \quad \rightarrow \quad \Delta x \dfrac{\partial}{\partial x} \sin$.
\end{enumerate} \vskip 5pt
which finally gives a nonlinear PDE representation of \eqref{Kuramoto_ex}:
\begin{equation} \label{Kuramoto_PDE}
\frac{\partial \rho}{\partial t} = \Delta x \dfrac{\partial}{\partial x} \sin\left(\Delta x\dfrac{\partial \rho}{\partial x}\right).
\end{equation}

Speaking about the accuracy, the simplest possible PDE can be obtained by choosing orders of accuracy equal to 0 for each set of similar subgraphs. For more accurate equations it makes sense to specify the desired order of the equation $d$ and then get rid of all the terms which consist of composition of derivatives of combined order higher than $d$.

\section{Extensions}

Until now we discussed systems with nodes which were uniformly placed on the infinite 1D line and which had common space-independent dynamics. The method can be extended to include more classes of systems. 

Spatially invariant systems \cite{BASSAM} such as periodic ones can be tackled by choosing different index spaces. In the periodic case we can assume that the positions $x \in \Ss$ are placed on the unit circle and indices $i \in \Z/n\Z$ form a ring of integers modulo $n$, where $n$ is the number of states of the original ODE. Since any function on $\Ss$ can be mapped to a periodic function on $\R$, the analysis in Sections~\ref{SEC_method_linear} and~\ref{SEC_nonlinear} remain the same.

Further the time dependence can be introduced into system gains both in the ODE and in the PDE, where the continuation is performed independently for every fixed $t$. This allows to use the method for time-varying systems and switching networks. Also systems whose state is vector-valued can be continualized using the same finite differences based scheme, thus in the following we will assume that the state of a system is scalar.

In the following subsections we will explore how the method can be extended to include systems with several spatial dimensions, systems with space dependence or nonuniform placing and systems with boundaries. Finally we introduce a concept of PDE with index derivatives which can be applied to systems whose states coincide with the positions in space, for example particle systems.

\subsection{Multidimensional systems} \label{SEC_multidim}

Let a position of a node $\rho_i$ be described by $x_i \in \R^n$. Moreover, a node $\rho_i$ is referenced by a multi-index $i = (i_1, ..., i_n) \in \Z^n$. We assume that the position difference between two neighbour nodes $i = (i_1, ..., i_k, ..., i_n)$ and $i' = (i_1, ..., i_k + 1, ..., i_n)$ is 
$$
x_{i'} - x_i = (0, ..., \Delta x_k, ..., 0) \qquad \forall k \in \{1..n\},
$$
and that there exists a vector $\Delta x = (\Delta x_1, ..., \Delta x_k, ..., \Delta x_n).$

We shall describe the extension for linear systems, as nonlinear systems can be dealt with in an analogous way by the same computational graph concept as in Section \ref{SEC_nonlinear}. Thus we start from the linear system \eqref{linear_ODE}. For nonnegative multi-index $h$ we define an absolute value $|h| = \sum_{k=1}^n h_k$. Further, we define multi-index power $h$ of a vector $x$ as $x^h = \prod_{k=1}^n x_k^{h_k}$, with an assumption $0^0 = 1$. By Taylor series
\begin{equation}
\rho_{i+s_j} = \rho_i + \Sum_{|h|=1} s_j^h \Delta x^{h} \frac{\partial \rho_i}{\partial x^h} + \Sum_{|h|=2} \frac{s_j^h}{2} \Delta x^{h} \frac{\partial^2 \rho_i}{\partial x^h} + ...,
\end{equation}
and the continuation up to the order $d$ thus is
\begin{equation} \label{pde_LTI_ND}
\frac{\partial \rho}{\partial t} = \Sum_{\substack{h\in\Z^n_+ \\ |h| \le d}} c_h \frac{\Delta x^h}{|h|!}  \frac{\partial^{|h|} \rho}{\partial x^h}, \qquad c_h = \Sum_{j=1}^N a_j s_j^h.
\end{equation}
Using $\Ss$ instead of $\R$, more complex multidimensional spaces can be covered such as torus or cylinder.

\subsection{Space-dependent systems}

Let us now look on the linear system \eqref{linear_ODE} with one important difference: the system gains $a_j$, the shifts $s_j$ and the number of neighbours $N$ become space-dependent:
\begin{equation} \label{linear_space_dependent_ODE}
\dot \rho_i = \Sum_{j=1}^{N_i} a_{ij} \rho_{i+s_{ij}}.
\end{equation}
Notice that equation \eqref{linear_space_dependent_ODE} describes in fact any linear system.

Now, choosing a unique $d$ such that $d + 1 \ge N_i$ for all $i$, one can perform a continuation \eqref{Continuation_LS} at every point $x_i$ up to the order $d$ and obtain a PDE \eqref{linear_PDE} with space dependent gains $c_{ik}$. This means that we know the gains $c_{ik}$ at the points with coordinates $x_i$, which can be seen as a sampling of some function $c_k(x)$ at points $x_i$. 

We can now perform either an interpolation or an approximation based on this sampling. In the first case we seek for $c_k(x)$ such that $c_k(x_i) = c_{ik}$, while in the second case it is enough to satisfy this relation approximately, for example by determining $c_k(x)$ via least squares. In either case, the resulting continuation of \eqref{linear_space_dependent_ODE} is given by
\begin{equation} \label{linear_space_dependent_PDE}
\frac{\partial \rho}{\partial t} = \Sum_{k=1}^{d} c_k(x) \frac{\Delta x^k}{k!} \frac{\partial^k \rho}{\partial x^k}.
\end{equation}

In the case of nonlinear systems the continuation can be performed if the computational graphs for every node compute the same dynamics. We can formalize it by stating the following property:

\begin{definition}
We say that two computational graphs \textit{have the same structure} if
\begin{enumerate}
\item their root nodes compute the same expression,
\item any child subgraph of the root node of the first graph \textit{has the same structure} with some child subgraph of the root node of the second graph and vice versa.
\end{enumerate}  
\end{definition}
This definition, formulated through recursion, essentially means that the order of nonlinearities which is hidden in two computational graphs should coincide, see Fig.~\ref{FIG_graph_def6}.

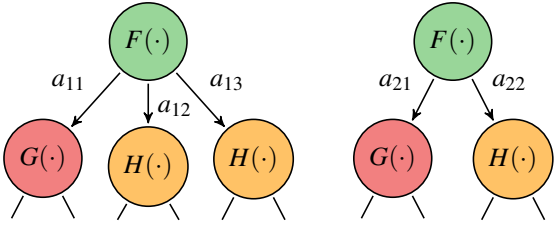
\begin{figure}[h]
\begin{center}
\begin{tikzpicture}[modal]
	
	\node[circle,draw,minimum size=1.0cm,fill=green!60!black!40] (f1) {$F(\cdot)$};
	
	\node[circle,draw,minimum size=1.0cm,fill=green!60!black!40] (f2) [right=3.0cm of f1, xshift=-0.0cm] {$F(\cdot)$};
	
	\node[circle,draw,minimum size=1.0cm,fill=red!90!black!50] (g1) [below=0.5cm of f1, xshift=-1.4cm] {$G(\cdot)$};
	\node[circle,draw,minimum size=1.0cm,fill=red!40!yellow!60] (h1) [below=0.6cm of f1, xshift=0.0cm] {$H(\cdot)$};
	\node[circle,draw,minimum size=1.0cm,fill=red!40!yellow!60] (hh1) [below=0.5cm of f1, xshift=1.4cm] {$H(\cdot)$};
	\node[circle,draw,minimum size=1.0cm,fill=red!90!black!50] (g2) [below=0.5cm of f2, xshift=-0.8cm] {$G(\cdot)$};
	\node[circle,draw,minimum size=1.0cm,fill=red!40!yellow!60] (h2) [below=0.5cm of f2, xshift=0.8cm] {$H(\cdot)$};
	
	\path[->] (f1) edge node[above left]{$a_{11}$} (g1);
	\path[->] (f1) edge node[above right, pos=0.99]{$a_{12}$} (h1);
	\path[->] (f1) edge node[above right]{$a_{13}$} (hh1);
	\path[->] (f2) edge node[above left]{$a_{21}$} (g2);
	\path[->] (f2) edge node[above right]{$a_{22}$} (h2);
	
	\node (g11) [below=0.3cm of g1, xshift=-0.5cm] {};
	\node (g12) [below=0.3cm of g1, xshift=0.5cm] {};
	\node (h11) [below=0.2cm of h1, xshift=-0.5cm] {};
	\node (h12) [below=0.2cm of h1, xshift=0.5cm] {};
	\node (hh11) [below=0.3cm of hh1, xshift=-0.5cm] {};
	\node (hh12) [below=0.3cm of hh1, xshift=0.5cm] {};
	\node (g21) [below=0.3cm of g2, xshift=-0.5cm] {};
	\node (g22) [below=0.3cm of g2, xshift=0.5cm] {};
	\node (h21) [below=0.3cm of h2, xshift=-0.5cm] {};
	\node (h22) [below=0.3cm of h2, xshift=0.5cm] {};
	
	\path[-] (g1) edge (g11);
	\path[-] (g1) edge (g12);
	\path[-] (h1) edge (h11);
	\path[-] (h1) edge (h12);
	\path[-] (hh1) edge (hh11);
	\path[-] (hh1) edge (hh12);
	\path[-] (g2) edge (g21);
	\path[-] (g2) edge (g22);
	\path[-] (h2) edge (h21);
	\path[-] (h2) edge (h22);

	
	
\end{tikzpicture}
\end{center}
\vskip -15pt
\caption{Illustration of two computational graphs having the same structure.} \label{FIG_graph_def6}
\end{figure}

Finally, a continuation of a nonlinear ODE system can be performed if all the computational graphs computing the dynamics for all states $\rho_i$ have the same structure. Indeed, in this case it is possible to perform a continuation for any set of similar subgraphs for each node as in the linear case of \eqref{linear_space_dependent_ODE}-\eqref{linear_space_dependent_PDE}. Moreover, by Definition 6 these sets of similar subgraphs for different positions serve as inputs to the same nonlinearities, therefore a unique PDE with space-dependent coefficients can be obtained.

\begin{remark}
In theory, it is possible to satisfy Definition 6 for any nonlinear system formulated through computational graphs. Indeed, assume two computational graphs have two different root node expressions, denoted as $F(\cdot)$ and $G(\cdot)$ respectively. Then we can artificially create a new common root node which will compute $1 \cdot F(\cdot) + 0 \cdot G(\cdot)$ for the first graph and $0 \cdot F(\cdot) + 1 \cdot G(\cdot)$ for the second. Thus we can satisfy the first condition of Definition 6, and recursively applying this idea one can transform any pair of computational graphs into the pair which has the same structure.
However, if the computational graphs of the system are too different in different points, it can make no sense to represent a system as a PDE, since it means that the dynamics of different parts of the system has nothing in common.
\end{remark}

\subsection{Unequally-spaced systems}

In the original derivation of \eqref{linear_combination_continuation} we assumed that the nodes are separated by constant $\Delta x$ in space. This was done for simplicity, and in general systems with nonuniform spacing can be tackled completely in the same way as space-dependent systems described in previous subsection. Indeed, assume the linear system is given by $\dot \rho_i = \sum_{j=1}^{N_i} a_{ij} \rho_{i+s_{ij}}$, where $x_{i+1} - x_i \not = x_{j+1} - x_j$ for $i \not = j$ in general. Then for every point the continuation can be performed by defining $c_{ik} = \sum_{j=1}^{N_i} a_{ij} (x_{i+s_{ij}}-x_i)^k$. Representing obtained gains by functions $c_k(x)$ as it was described in previous subsection one finishes with a PDE
\begin{equation} \label{linear_space_step_dependent_PDE}
\frac{\partial \rho}{\partial t} = \Sum_{k=1}^{d} c_k(x) \frac{1}{k!} \frac{\partial^k \rho}{\partial x^k}.
\end{equation}

\subsection{Boundary conditions} \label{SEC_boundary}

Now let us look at the Heat PDE:
\begin{equation} \label{Heat_PDE}
\frac{\partial \rho}{\partial t} = \frac{\partial^2 \rho}{\partial x^2}.
\end{equation}
Imagine that this equation is defined on an interval $x~\in~[0,~+\infty)$, that is there is a boundary in the point $x = 0$. 

There are two types of boundary conditions (or BC) which can be supplied to provide a well-posed boundary value problem. For example for some $a \in \R$,
\begin{equation} \label{Heat_PDE_BC}
\begin{aligned}
\text{1)} \quad &\textit{Dirichlet} \text{ BC: } \rho(0) = a, \\
\text{2)} \quad &\textit{Neumann} \text{ BC: } \partial\rho/\partial x \: (0) = a.
\end{aligned}
\end{equation}
There can also exist a linear combination of these boundary conditions, called \textit{Robin} BC.

If the Heat Equation \eqref{Heat_PDE} is discretized in stencil points $\{i-1, i, i+1\}$, the result is
\begin{equation} \label{Heat_ODE}
\dot \rho_i = \frac{1}{\Delta x^2} \left( \rho_{i-1} - 2 \rho_i + \rho_{i+1} \right).
\end{equation}
Assume now that there exists $i_0 = 1$ such that $x_{i_0-1} = 0$. Depending on the type of boundary conditions, the equation for the state $\rho_{1}$ can be obtained by the discretization of a boundary value problem \eqref{Heat_PDE}-\eqref{Heat_PDE_BC} in two ways:
\begin{equation}\label{Heat_ODE_BC}
\begin{aligned}
\text{1)} \quad &\text{Dirichlet BC: } \dot \rho_{1} = \left( a - 2 \rho_{1} + \rho_{2} \right) / \Delta x^2, \\
\text{2)} \quad &\text{Neumann BC: } \dot \rho_{1} = \left( \rho_{2} - \rho_{1} \right) / \Delta x^2 - a/ \Delta x.
\end{aligned}
\end{equation}

Now imagine the system \eqref{Heat_PDE} is obtained by the continuation process from the system \eqref{Heat_ODE}. We can notice that states of \eqref{Heat_ODE} are governed by the same dynamics except for the boundary state $\rho_{1}$. The question is how to recover the boundary conditions \eqref{Heat_PDE_BC} for the PDE from the dynamics of $\rho_{1}$ in \eqref{Heat_ODE_BC}.

This indeed can be done if one assumes that there exists a ''ghost cell'' $\rho_{0}$ such that it has no dynamics, but is algebraically connected with adjacent states. With a proper definition of $\rho_{0}$ the equation for $\dot \rho_{1}$ can be represented in the same way as for other states \eqref{Heat_ODE} and thus has the same continuation \eqref{Heat_PDE}. For example, algebraic equations for $\rho_{0}$ representing \eqref{Heat_ODE}-\eqref{Heat_ODE_BC} are
\begin{equation}\label{Heat_ODE_BC_continuation}
\begin{aligned}
\text{1)} \quad &\text{Dirichlet BC: } \rho_{0} = a, \\
\text{2)} \quad &\text{Neumann BC: } \rho_{0} = \rho_{1} - a \Delta x.
\end{aligned}
\end{equation}
The ghost cell $\rho_0 = a$ for the case of Dirichlet BC is depicted in Fig.~\ref{FIG_boundary}. Notice that equations \eqref{Heat_ODE_BC_continuation} can be directly continualized, obtaining \eqref{Heat_PDE_BC}.

\begin{figure}[h]
\begin{center}
\begin{tikzpicture}[modal]
	
	\node[circle,draw,minimum size=1.1cm,fill=red!40!yellow!60] (x0) {$\rho_{1}$};
	\node[circle,draw,minimum size=1.1cm,fill=red!40!yellow!60] (x1) [right=0.3cm of x0] {$\rho_{2}$};
	\node[circle,draw,minimum size=1.1cm,fill=gray!15] (ghost) [left=0.3cm of x0] {$a$};
	\node (p0) [right=0.3cm of x1] {};
	\node[point_s] (p1) [right=0.5cm of x1] {};
	\node[point_s] (p2) [right=0.7cm of x1] {};
	\node[point_s] (p3) [right=0.9cm of x1] {};
	
	\path[-] (x0) edge (x1);
	\path[-] (x1) edge (p0);
	\path[-] (x0) edge (ghost);
	
	\path[->] ($(x0)+(-2.5,-1.1)$) edge node[pos=1,right]{$x$} ($(x0)+(4.5,-1.1)$);
	
	\path[-] ($(x0)+(-1.4,-1.2)$) edge node[below]{$0$} ($(x0)+(-1.4,-1.0)$);
	\path[-] ($(x0)+(0,-1.2)$) edge node[below]{$\Delta x$} ($(x0)+(0,-1.0)$);
	\path[-] ($(x0)+(1.4,-1.2)$) edge node[below]{$2 \Delta x$} ($(x0)+(1.4,-1.0)$);
	
\end{tikzpicture}
\end{center}
\vskip -10pt
\caption{Boundary of the system \eqref{Heat_ODE} with Dirichlet boundary condition \eqref{Heat_ODE_BC}, represented by a ghost cell $\rho_0 = a$.} \label{FIG_boundary}
\end{figure}
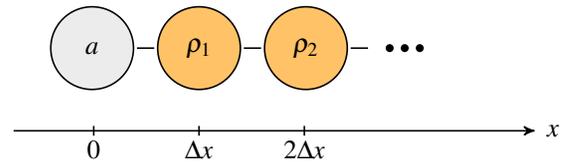

This procedure can be generalized to any ODE system: once the states near boundaries change their dynamics with respect to the general governing equation, this change can be represented by ''ghost cells'' with algebraic dependences on the ''real'' states. Continualizing these algebraic equations leads to the boundary conditions for the obtained PDE.

\subsection{PDE with index derivatives} \label{SEC_pde_index}

Usually PDEs have derivatives written with respect to the time and space variables, thus their physical meaning is in the function continuously varying in time and space. However, in general no one prevents us from writing a PDE with respect to some other variables. 

Assume a physical system is given by a set of interacting agents, with agents being indexed by an integer index $i \in \Z$ (a general multiindex space $\Z^n$ can also be used). Let an agent $i$ have a state $\rho_i$. The index variable $i$ is by definition discrete. However we can make an assumption that in between of two agents with consecutive indexes $i$ and $i+1$ there is a continuum of virtual agents having state varying from $\rho_i$ to $\rho_{i+1}$. Denoting this continuously varying index by $M \in \R$ we can say that the state of the system $\rho$ is a continuous and smooth function $\rho(M)$ with the property $\rho(i) = \rho_i$. This definition of $M$ coincides with the definition of Moskowitz function used to describe the number of vehicles passed through a fixed point in traffic modeling \cite{EE_MOSKOWITZ}.

Once the index variable is continuous, we can think about it as a new space variable. Thus it is possible to use a continuation described in previous sections, where the distance between two consecutive agents is obviously $\Delta M = 1$. The derivatives of the state with respect to the index can be obtained by continuation, for example $\rho_{i+1} - \rho_i \; \to \; \partial \rho / \partial M$. 

\section{Derivation of the Euler Equations}

In the beginning of the XX century Hilbert posed his 6th problem, where he suggested to develop a rigorous way leading from the atomistic view to the laws of motion of continua. In particular, the problem can be formulated as a derivation of the Euler equations for compressible fluids from the Newton's dynamics of individual particles. 

For the most famous case of particles interacting through collision the Boltzmann equation was developed, describing evolution of the joint position-velocity probability distribution of particles. The method of how to transform individual's dynamics into Boltzmann equation is based on the Boltzmann-Grad limit \cite{EE_BOOK_BZ}, assuming velocities of colliding particles being independent. The following transformation from the Boltzmann equation to the Euler equations uses either Hilbert or Chapman-Erskog expansions with space contration limits \cite{EE_BOOK_EXP,EE_CHAPMAN}, Grad moments \cite{EE_GRAD} or the method of invariant manifolds \cite{EE_REVIEW}.

Another situation arises when the particles interact through long-range forces. In this case the Vlasov equation can be used instead of the Boltzmann equation to describe the joint position-velocity probability distribution. The derivation of the Euler equations from the Vlasov equation was performed in \cite{EE_VLASOV_EULER} using space-contracting limit. In particular it was shown that the resulting system has zero temperature, i.e. the velocities of individual particles coincide with the velocity field. However, due to the space contraction the particular form of the potential function was lost and the obtained pressure was just a square of the density.

Here we present a derivation of Euler equations directly from individual's dynamics using the continuation method described in previous sections. Contrary to other works, we do not use any kind of limits and we use only one assumption on the isotropy of the space. The assumption requires that for any particle its nearest neighbours are distributed around uniformly in every direction, which can be seen as a counterpart to the molecular chaos hypothesis for the standard derivation of the Boltzmann equation.

\subsection{System of particles}

It is assumed that the fluid consists of small particles interacting with each other, with every particle following simple Newton laws. We will study the system with $n$ space dimensions, and the particles are assumed to have unit mass. 

We further assume that there is an interaction between each pair of particles which is given by a force
\begin{equation} \label{Euler_force}
F(x_i - x_j) = \frac{x_i - x_j}{\Vert x_i - x_j\Vert} f(\Vert x_i - x_j\Vert) = (x_i - x_j) \phi(\Vert x_i - x_j\Vert),
\end{equation}
thus the force acts along the line connecting two particles with the magnitude $f$ depending only on the distance between particles. For simplicity we also define a function $\phi(s) = f(s)/s$ representing the scaled magnitude. Since we assume the infinite number of particles and an infinitely large space, the magnitude of the force should satisfy
\begin{equation} \label{Euler_force_condition}
\Int_{\eps}^{+\infty} s^{n - 1} f(s) ds < \infty \quad \forall \eps > 0
\end{equation}
such that the cumulative force on any particle is finite. Thus the interaction should be fast-decaying.

We then need to enumerate all particles. For this we will use multiindex $i \in \Z^n$ as written in section \ref{SEC_multidim}. Now let us write the dynamics of a particle with multiindex $i$ using the second Newton's Law:
\begin{equation} \label{Euler_ODE}
\left\{
\begin{aligned}
\dot x_i &= v_i, \\
\dot v_i &= \Sum_{q \not = 0} F(x_i - x_{i+q}),
\end{aligned}
\right.
\end{equation}
where the summation is performed among all multiindices $q$ in $\Z^n\setminus\{0\}$, since all the particles interact with each other. Both the position $x_i$ and the velocity $v_i$ are vectors in $\R^n$. 

\subsection{Derivation in the Euclidean space}

Treating the coordinate $x_i$ as a state and using the idea written in section \ref{SEC_pde_index} we define a multiindex function $M(x,t)$ which is the inverse function of the coordinate: $M(x_i,t) := i$. Likewise, $x(i,t) = x_i(t)$ and thus $x(M(x,t),t) \equiv x \;\; \forall x \in \R^n$. 

Now let us write a property of inverse function of multiindex as $M(x(M,t),t) \equiv M \;\; \forall M \in \R^n$, where the space for multiindices is continuous by the assumption in section \ref{SEC_pde_index}. Taking the time and the index derivatives, we obtain the following very useful relations on Jacobians:
\begin{equation} \label{Euler_M_t}
\frac{\partial M}{\partial t} + \frac{\partial M}{\partial x} \frac{\partial x}{\partial t} = 0,
\end{equation}
\begin{equation} \label{Euler_M_x}
\frac{\partial M}{\partial x} \frac{\partial x}{\partial M} = I.
\end{equation}

Equation \eqref{Euler_M_t} can be seen as a PDE where the function $M$ depends both on $x$ and $t$. Recalling that the multiindex in assumed to be continuous, we can further utilize the first equation of \eqref{Euler_ODE} written in a form $\dot x(M,t) = v(M,t)$, substitute it in \eqref{Euler_M_t} and obtain the following equation on the multiindex evolution:
\begin{equation} \label{Euler_M_PDE}
\frac{\partial M}{\partial t} = -\frac{\partial M}{\partial x} v(M(x,t)) = -\frac{\partial M}{\partial x} u(x,t),
\end{equation}
where the velocity function $u(x,t) = v(M(x,t),t)$ is defined as a velocity of a particle at some given point in space. Finally, taking the derivative with respect to space, we obtain
\begin{equation} \label{Euler_Mx_PDE}
\frac{\partial}{\partial t} \left( \frac{\partial M}{\partial x} \right) = -\frac{\partial}{\partial x} \left( \frac{\partial M}{\partial x} u \right).
\end{equation}

The Jacobian matrix $\frac{\partial M}{\partial x}(x,t)$ represents a \textit{compression tensor}, which measures how close are neighbour particles with respect to different directions in the euclidean space. Evolution of this Jacobian in the euclidean space is described by the matrix PDE \eqref{Euler_Mx_PDE}, which is essentially a transport equation with flow velocity given by $u(x,t)$. 

Now we approach the second equation in \eqref{Euler_ODE}. It would be desirable to transform it in such a way that we could obtain an evolution equation for the flow velocity $u(x,t)$. First of all, let us rewrite the second equation of \eqref{Euler_ODE} in a way more suitable for continuation, namely
\begin{equation} \label{Euler_v_ODE}
\dot v_i = - \Sum_{q > 0} \left(F(x_{i+q} - x_i) - F(x_i - x_{i-q}) \right),
\end{equation}
where the summation is performed among all multiinidices which are greater than zero in lexicographical order, i.e. the first nonzero element of $q$ should be positive.

Now we use the continuation of order of accuracy 0 on a multidimensional system as described in \ref{SEC_multidim} such that 
$$
x_{i+q} - x_i \; \rightarrow \; \frac{\partial x}{\partial M} \left(x_{i+q/2} \right) q, \quad  x_i - x_{i-q} \; \rightarrow \; \frac{\partial x}{\partial M} \left(x_{i-q/2} \right) q,
$$
which means that \eqref{Euler_v_ODE} becomes 
$$
\dot v_i = - \Sum_{q > 0} \left(F\left(\frac{\partial x}{\partial M} q\right)_{i+q/2} - F\left(\frac{\partial x}{\partial M} q\right)_{i-q/2} \right).
$$
Applying the continuation further to the forces, we obtain
$$
F_{i+q/2} - F_{i-q/2} \quad \rightarrow \quad \frac{\partial F}{\partial M} (x_i) q.
$$
Thus \eqref{Euler_v_ODE} transforms into
\begin{equation} \label{Euler_v_PDE_1}
\frac{\partial v}{\partial t} = -\Sum_{q>0} \frac{\partial}{\partial M} \left(\left[ \frac{\partial x}{\partial M} q \right] \phi\left(\left\Vert\frac{\partial x}{\partial M} q \right\Vert \right) \right) q,
\end{equation}
where we used a definition of the force \eqref{Euler_force}.

Now, we state the following result:
\begin{proposition}
For any $q \in \Z^n$ and for any smooth scalar field $\phi(x)$ the following identity holds:
\begin{equation} \label{Prop_1}
\begin{aligned}
&\left[\frac{\partial}{\partial M} \left( \phi \frac{\partial x}{\partial M} q \right) q \right]^T = \\
&\nabla \cdot \left( \phi\frac{\partial x}{\partial M} q q^T \frac{\partial x}{\partial M}^T\right) - \phi \left( \nabla \cdot \left(\frac{\partial x}{\partial M} \right) q q^T \frac{\partial x}{\partial M}^T\right),
\end{aligned}
\end{equation}
where $\nabla$ denotes a row vector of derivatives with respect to $x$.
\end{proposition}
\begin{proof}
First, for convenience denote the left-hand side as a vector $Q$:
\begin{equation}
Q := \frac{\partial}{\partial M} \left( \phi \frac{\partial x}{\partial M} q \right) q = \frac{\partial}{\partial x} \left( \phi \frac{\partial x}{\partial M} q \right) \frac{\partial x}{\partial M} q.
\end{equation}
Also define $h = (\partial x/\partial M) q$. Expanding $\partial (\phi h)/\partial x$, we get
\begin{equation} \label{Prop_1_Q}
Q = h \frac{\partial \phi}{\partial x} h + \frac{\partial h}{\partial x} h \phi = h h^T \frac{\partial \phi}{\partial x}^T + \frac{\partial h}{\partial x} h \phi.
\end{equation}
Now, for any $h \in \R^n$
$$
\nabla \cdot (h h^T) = \begin{pmatrix}
\Sum_i h_1 \frac{\partial h_i}{\partial x_i} + \Sum_i h_i \frac{\partial h_1}{\partial x_i} & \cdots & \Sum_i h_n \frac{\partial h_i}{\partial x_i} + \Sum_i h_i \frac{\partial h_n}{\partial x_i}
\end{pmatrix},
$$
which means that 
\begin{equation}
\left(\nabla \cdot (h h^T)\right)^T = \frac{\partial h}{\partial x} h + (\nabla \cdot h) h.
\end{equation}
Therefore the transpose of \eqref{Prop_1_Q} is
\begin{equation} \label{Prop_1_Q_T}
Q^T = \frac{\partial \phi}{\partial x} h h^T + \phi \nabla \cdot (hh^T) - \phi (\nabla \cdot h) h^T.
\end{equation}
Since for any matrix $J$ and for any scalar field $\alpha$
\begin{equation} \label{Prop_1_scalar}
\nabla \cdot (\alpha J) = \frac{\partial \alpha}{\partial x} J + \alpha \nabla \cdot J,
\end{equation}
we can simplify \eqref{Prop_1_Q_T} as
\begin{equation}
Q^T = \nabla \cdot (\phi hh^T) - \phi (\nabla \cdot h) h^T.
\end{equation}
The result of the proposition follows by substituting $h$ and noticing that $\nabla \cdot ((\partial x/\partial M) q) = (\nabla \cdot (\partial x/\partial M)) q$.
\end{proof}

Proposition 1 allows us to rewrite \eqref{Euler_v_PDE_1} as being dependent only on the euclidean space divergences and the inverse of the compression tensor $\partial M/\partial x$. To finalize the derivation of  a complete set of equations, recall the definition of the velocity field $u(x,t) = v(M(x,t),t)$. Taking the time derivative:
$$
\frac{\partial u}{\partial t} = \frac{\partial v}{\partial t} + \frac{\partial v}{\partial M} \frac{\partial M}{\partial t},
$$
which by \eqref{Euler_M_PDE} is
$$
\frac{\partial u}{\partial t} = -\frac{\partial v}{\partial M} \frac{\partial M}{\partial x} u + \frac{\partial v}{\partial t}.
$$
This equation can be simplified by $\partial u /\partial x = \partial v /\partial M \cdot \partial M /\partial x$. Finally, substituting \eqref{Euler_v_PDE_1} and \eqref{Prop_1} and combining the result with \eqref{Euler_Mx_PDE} we obtain a system
\begin{equation} \label{Euler_general}
\left\{
\begin{aligned}
\frac{\partial}{\partial t} \left( \frac{\partial M}{\partial x} \right) = &-\frac{\partial}{\partial x} \left( \frac{\partial M}{\partial x} u \right), \\
\frac{\partial u}{\partial t} = &-\frac{\partial u}{\partial x} u - \Sum_{q>0} \Bigg[ \nabla \cdot \left( \phi\frac{\partial x}{\partial M} q q^T \frac{\partial x}{\partial M}^T\right) \\ &- \phi \left( \nabla \cdot \left(\frac{\partial x}{\partial M} \right) q q^T \frac{\partial x}{\partial M}^T\right) \Bigg]^T,
\end{aligned} \right.
\end{equation}
where $\phi = \phi(\norm{(\partial x/\partial M)q})$. 

The system \eqref{Euler_general} has 12 states in 3-dimensional space, 9 for $\partial M / \partial x \:(x,t)$ and 3 for $u(x,t)$. It resembles the famous Grad 13-moment system \cite{EE_GRAD}, which extends the Euler equations by considering directional-dependent pressure tensor. The last state of the Grad 13-moment system is the inner energy, which does not appear in \eqref{Euler_general}. The reason for this is that we derive a continuous interaction term explicitly from the interaction forces, which is possible only if the forces are defined by long-range potentials. As it was shown in \cite{EE_VLASOV_EULER}, expressing a system with long-range potentials by the Euler equations leads to the solution with zero temperature, therefore the inner energy becomes functionally dependent on the velocity field and its evolution equation can be omitted.

\subsection{Dimensionality reduction}

It appears that in some special cases it is possible to reduce the system \eqref{Euler_general} by considering only one scalar characteristic of a compression in any space point instead of the whole compression tensor. 

Indeed, we define a \textit{density} as a determinant of the compression tensor, $\rho(x,t) := \det \left(\partial M / \partial x \right) (x,t)$. Not only the compression tensor itself, but also its determinant satisfies \eqref{Euler_Mx_PDE}. This nontrivial fact holds because the compression tensor is the Jacobian, and the proof is given in Lemma 1 in the Appendix. Therefore from \eqref{Euler_Mx_PDE}
\begin{equation} \label{Euler_rho_PDE}
\frac{\partial \rho}{\partial t} = -\nabla \cdot \left( \rho u \right).
\end{equation}
This equation is the first of the complete set of Euler equations. 

Unfortunately, the second equation of \eqref{Euler_general} depends on the whole compression tensor and thus it cannot be described only by the means of density. This is reasonable since in general the system can have different pressures in different directions in response to different compressions. Therefore in order to simplify the system we need to assume that the compression can be represented by a single number, i.e. that it is compressed equally in all directions.

\begin{assumption}[\textit{Isotropy}]
Compression tensor $\partial M / \partial x (x,t)$ is isotropic (equal in all directions), thus it can be represented as a rotation matrix multiplied by a scalar. 
\end{assumption}

This assumption looks restricting at first glance, but for the infinitely large system with infinitely many particles the system indeed ''looks the same'' in all directions at every point, thus we can say it is isotropic. 

Assumption 1 has long-lasting implications. Define $l(x,t) := \lambda(\partial x/\partial M (x,t))$, since all the eigenvalues are equal. This variable, called \textit{specific distance}, represents an average distance between two neighbouring particles at point $x$. By definition of the density $\rho = l^{-n}$. Further, $\left\Vert\frac{\partial x}{\partial M} q \right\Vert = l \norm{q}$. Breaking the summation in \eqref{Euler_general} in a sum of all possible lengths $r$ of multiindex vectors, we can rewrite the summation term as
\begin{equation} \label{Isotropy_summation}
\begin{aligned}
\Sum_{r^2 \in \N} \Bigg[ &\nabla \cdot \left( \phi(rl)\frac{\partial x}{\partial M} \Sum_{\substack{q > 0 \\ \norm{q} = r}} \left(q q^T \right) \frac{\partial x}{\partial M}^T\right) - \\ &- \phi(rl) \left( \nabla \cdot \left(\frac{\partial x}{\partial M} \right) \Sum_{\substack{q > 0 \\ \norm{q} = r}} \left(q q^T \right) \frac{\partial x}{\partial M}^T\right) \Bigg]^T.
\end{aligned}
\end{equation}

\begin{proposition}
Given $r$ such that $r^2 \in \N$, the summation over all outer products of multiindices of a length $r$ is proportional to the identity matrix, i.e. there exists $\beta(r)$ such that
\begin{equation} \label{Euler_qq}
\Sum_{\substack{q > 0 \\ \norm{q} = r}} q q^T = \beta(r) I.
\end{equation}
\end{proposition}
\begin{proof}
First of all, we will show that all nondiagonal elements in \eqref{Euler_qq} are zero. Indeed, for any positive $q$ its contribution to $kj$-th element of matrix \eqref{Euler_qq} is given by $q_k q_j$. But for any $k \not = j$ we can pick $\bar q$ such that it equals $q$ except $\bar q_{\max(k,j)}~=~-q_{\max(k,j)}$. In this case $\bar q$ is also positive and thus is included into the summation, while the contribution to $kj$-th element of \eqref{Euler_qq} has opposite sign. Therefore all nondiagonal elements of \eqref{Euler_qq} are zero.

Further, all diagonal elements of \eqref{Euler_qq} are equal. This can be proven by analogous argument. Indeed, we can take a positive $q$ and look at the elements $q_k^2$ and $q_j^2$. Then $\bar q$ which is equal to $q$ except for $\bar q_k = \sgn(q_k) |q_j|$ and $\bar q_j = \sgn(q_j) |q_k|$ is also positive, but swaps the contributions between $k$-th and $j$-th diagonal elements. Thus all the contributions to the diagonal elements are equal.
Finally, 
\begin{equation} 
\trace \Sum_{\substack{q > 0 \\ \norm{q} = r}} q q^T = \Sum_{\substack{q > 0 \\ \norm{q} = r}} q^T q = r^2 \cdot \#_r q = n \beta(r),
\end{equation}
where $\#_r q$ denotes the number of positive multiindices $q$ with length $r$ and we define $\beta(r) = r^2 / n \cdot \#_r q$.
It is worth noticing that by \cite{EE_SPHERE} the average approximate behaviour of the number of positive multiindices $q$ with length $r$ is $\#_r q \propto r^{n - 1}$ as $r \to +\infty$, thus $\beta(r) \propto r^{n}$.
\end{proof}

By Assumption 1 
\begin{equation} \label{Isotropy_tensor}
\frac{\partial x}{\partial M} \frac{\partial x}{\partial M}^T = l^2 I.
\end{equation}
Using Proposition 2 and \eqref{Isotropy_tensor}, \eqref{Isotropy_summation} becomes
$$
\Sum_{r^2 \in \N} \beta(r) \Bigg[ \nabla \cdot \left( \phi(rl)l^2 I\right) - \phi(rl) \left( \nabla \cdot \left(\frac{\partial x}{\partial M} \right) \frac{\partial x}{\partial M}^T\right) \Bigg]^T.
$$
The value inside of the square brackets can be simplified further. Indeed, by \eqref{Prop_1_scalar} it is possible to inject density inside, which gives
$$
\footnotesize{
\begin{aligned}
\frac{1}{\rho}\Bigg[ &\nabla \cdot \left( \rho\phi(rl)l^2 I\right) - \frac{\partial \rho}{\partial x} \phi(rl)l^2 I - \phi(rl) \rho \left( \nabla \cdot \left(\frac{\partial x}{\partial M} \right) \frac{\partial x}{\partial M}^T\right) \Bigg]^T \\ &= \frac{1}{\rho}\Bigg[ \nabla \cdot \left( \rho\phi(rl)l^2 I\right) - \phi(rl) \left( \nabla \cdot \left(\rho \frac{\partial x}{\partial M} \right) \frac{\partial x}{\partial M}^T\right) \Bigg]^T.
\end{aligned}}
$$
Finally, the second term in the square brackets appears to be zero, since Lemma 2 in the Appendix proves that $\nabla~\cdot~\left(\rho\frac{\partial x}{\partial M}\right)~\frac{\partial x}{\partial M}^T~=~0$. Using this Lemma and the fact that $\nabla \cdot (\rho\phi(rl)l^2 I) = \nabla (\rho\phi(rl)l^2)$, we can define the \textit{pressure}:
\begin{equation} \label{Euler_P}
P = \Sum_{r^2 \in \N} \beta(r) \rho \phi(l r) l^2 = \Sum_{r^2 \in \N} \frac{\beta(r)}{r} l^{1-n} f(l r).
\end{equation}
Note that the pressure is well-defined since the sum is convergent by the  property \eqref{Euler_force_condition}. With this definition, the system \eqref{Euler_general} together with \eqref{Euler_rho_PDE} turns into the famous \textit{Euler equations}:
\begin{equation} \label{Euler_PDE}
\left\{
\begin{aligned}
\frac{\partial \rho}{\partial t} &= -\nabla \cdot \left( \rho u \right), \\
\frac{\partial u}{\partial t} &= -\frac{\partial u}{\partial x} u - \frac{\nabla P}{\rho}^T.
\end{aligned}
\right.
\end{equation}
Therefore the following theorem was proven:
\begin{theorem}
There exists a valid continuation process which leads from the Newtonian system \eqref{Euler_ODE} to the Euler equations \eqref{Euler_PDE} under the assumption that the system is locally isotropic in every point in space.
\end{theorem}

\begin{remark}[Non-complete interaction topologies]
In the original ODE system \eqref{Euler_ODE} we assumed that an interaction exists between every pair of particles, i.e. that the topology of interactions is all-to-all. In general in order to obtain \eqref{Euler_ODE} it would be sufficient to use any topology for which the isotropy required in Assumption 1 is possible. The difference in topologies would modify the definitions of density $P(x,t)$ in \eqref{Euler_P}. 

For example, for the grid topology with equations given by
\begin{equation} \label{Euler_ODE_grid}
\left\{
\begin{aligned}
\dot x_i &= v_i, \\
\dot v_i &= \Sum_{k=1}^n \left( F(x_i - x_{i-e_k}) - F(x_i - x_{i+e_k}) \right),
\end{aligned}
\right.
\end{equation}
where $e_k$ denotes the $k$-th basis vector of $\R^n$, the continuation renders the same Euler equations \eqref{Euler_PDE} with the pressure given by $P = f(l) / l^{n-1}$.
\end{remark}

\section{Control of Robotic Swarm}

In this section we will demonstrate how the continuation method described above can help in the analysis and design of control laws for large-scale systems. We will do it by using an example of a robotic swarm, i.e. a formation of robots whose goal is to follow some desired trajectory while passing through obstacles and preserving relative agents' positions.

Control of robotic formations is an extensively studied topic, see recent reviews \cite{RS_REVIEW15, RS_REVIEW18}. However most of the methods rely on the graph-theoretic properties of interaction topology and on simple linear controllers to provide stability. A PDE approach was taken in \cite{RS_BIRDS95} where the Euler PDE with diffusion terms was used to model the flocks of birds. The authors proposed a PDE to describe the behaviour of agents and analyzed it to study a symmetry breaking which leads to a coherent movement of birds. Similar PDE model was used to control 3D agent formation with 2D disc communication topology via backstepping in \cite{RS_KRSTIC}. Another interesting concept was used in \cite{RS_PdE06, RS_PdE08}, namely the \textit{Partial difference Equations} (PdEs), which are an analogue of ordinary Partial Differential Equations defined on graphs. With the help of Lyapunov analysis within PdE formalism it was proven that the linear diffusion controller stabilizes the formation.

Contrary to previous works we will base our analysis on the continuation procedure, rigorously introducing a PDE to describe a formation of drones. We will study this PDE and recover a nonlinear local control law which, being applied to the agents, forces the whole formation to follow the desired density profile.

\subsection{Continuation and PDE Control}

Let us start from a system of drones having double integrator dynamics:
\begin{equation} \label{Drones_ODE}
\ddot x_i = \tau_i.
\end{equation}
Here $x_i\in\R^n$ is a position of the $i$-th drone in $n$-dimensional space and $\tau_i\in\R^n$ is a control we want to design. The drones are enumerated with multiindices $i\in\Z^n$. Define $v_i = \dot x_i$. Similarly to the previous section we introduce multiindex function $M(x,t)$ such that $M(x_i,t) \equiv i$ and then perform a continuation. The resulting system is
\begin{equation} \label{Drones_PDE}
\left\{ 
\begin{aligned}
\frac{\partial \rho}{\partial t} &= -\nabla\cdot(\rho u), \\
\frac{\partial u}{\partial t} &= -\frac{\partial u}{\partial x} u + \tau(x,t), 
\end{aligned} \right.
\end{equation}
where $\tau(x,t) = \tau(M(x,t),t)$ is a continuation of the control $\tau_i$.

Now let us formulate the desired system which will be used as a reference the real formation should converge to. Given a velocity profile $u_d(x)$, we define the desired density $\rho_d(x,t)$ to follow this velocity profile. Essentially this means ''desired agents'' have single-integrator dynamics. Note that in general $u_d$ can be dependent on time but we don't consider it for simplicity of writing.

Thus we assume the desired system is governed by
\begin{equation} \label{Drones_des}
\frac{\partial \rho_d}{\partial t} = -\nabla\cdot(\rho_d u_d).
\end{equation}
Our goal is to derive $\tau(x,t)$ such that $\rho \to \rho_d$. First, direct calculations from \eqref{Drones_PDE} and \eqref{Drones_des} lead to the following systems in terms of flows $(\rho u)$ and $(\rho_d u_d)$:
\begin{equation}
\begin{aligned}
\frac{\partial (\rho u)}{\partial t} &= -\nabla\cdot(\rho u) u -\rho \frac{\partial u}{\partial x} u + \rho \tau(x,t), \\
\frac{\partial (\rho_d u_d)}{\partial t} &= -\nabla\cdot(\rho_d u_d) u_d.
\end{aligned}
\end{equation}
Define the deviation from the desired density $\tilde\rho = \rho - \rho_d$. Then the second-order equation for the deviation is
$$
\frac{\partial^2 \tilde\rho}{\partial t^2} = \nabla \cdot \left[ \nabla\cdot(\rho u) u - \nabla\cdot(\rho_d u_d) u_d + \rho \frac{\partial u}{\partial x} u - \rho \tau(x,t) \right].
$$
In order to cancel the nonlinear terms, define now the control $\tau$ as
\begin{equation} \label{Drones_control_continuous}
\begin{aligned}
\tau = \frac{\partial u}{\partial x} u + \frac{1}{\rho}\Big[ \nabla&\cdot(\rho u) u - \nabla\cdot(\rho_d u_d) u_d +{} \\ &{}+ \alpha (\rho_d u_d - \rho u) + \beta \nabla(\rho_d - \rho)^T \Big],
\end{aligned}
\end{equation}
where $\alpha$ and $\beta$ are some positive gains. Then the equation for the density deviation transforms into
\begin{equation} \label{Drones_wave}
\frac{\partial^2 \tilde\rho}{\partial t^2} = -\alpha \frac{\partial \tilde\rho}{\partial t} + \beta \nabla^2 \tilde\rho.
\end{equation}
This equation is a wave equation with damping and thus it is asymptotically stable if $\tilde\rho = 0$ on the boundary of the domain \cite{WAVE}. Choosing a desired system such that $\rho_d = 0$ on the boundary and using a continuation of $\rho$ such that $\rho = 0$ on the boundary ensures satisfaction of the boundary condition. 

\subsection{Discretization of the control}

Formula \eqref{Drones_control_continuous} for PDE \eqref{Drones_PDE} is local by its nature, but it should be discretized to be implemented on every agent of the original ODE \eqref{Drones_ODE}. One particular discretization is described next.

First of all, for the agent $i$ define a matrix $G_i$ as a discretization of the compression tensor:
\begin{equation} \label{Drones_discr_G}
[G_i]_j = (x_{i+e_j} - x_{i-e_j})/2 \approx \frac{\partial x}{\partial M_j}(x_i,t),
\end{equation}
where $e_j$ is the $j$-th unit basis vector and $[G_i]_j$ represent the $j$-th column of $G_i$. The matrix $G_i$ depends on the positions of $2n$ neighbouring agents of the $i$-th agent. In the same way as $G_i$ we define a matrix $W_i$ representing a velocity Jacobian:
\begin{equation} \label{Drones_discr_W}
[W_i]_j = (v_{i+e_j} - v_{i-e_j})/2 \approx \frac{\partial u}{\partial M_j}(x_i,t).
\end{equation}

Now we can write formulas for all terms inside of \eqref{Drones_control_continuous} depending on the real system:
\begin{equation} \label{Drones_discretizations}
\begin{aligned}
\text{1).}& \quad
\frac{\partial u}{\partial x} u = \frac{\partial u}{\partial M} \frac{\partial M}{\partial x} u \approx W_i G_i^{-1} v_i, \\
\text{2).}& \quad
\nabla \cdot u = \Sum_{j=1}^n \frac{\partial u_j}{\partial M} \frac{\partial M}{\partial x_j} \approx \Sum_{j=1}^n [W_i^T]_j \cdot [G_i^{-1}]_j, \\
\text{3).}& \quad
\rho \approx 1 / \det G_i , \\
\text{4).}& \quad
\nabla \rho \approx -\rho^2 \nabla (\det G_i) \approx -\rho^2 \frac{\partial (\det G_i)}{\partial M} G_i^{-1}, \\
\end{aligned}
\end{equation}
where the gradient of the determinant $\det G_i$ should be computed according to the determinant formula, using second derivatives of the positions discretized similarly to \eqref{Drones_discr_G}:
$$
\begin{aligned}
\frac{\partial^2 x}{\partial M_j\partial M_k}(x_i,t) &\approx (x_{i+e_j+e_k} + x_{i-e_j-e_k} - x_{i+e_j-e_k} - x_{i-e_j+e_k}) / 4, \\
\frac{\partial^2 x}{\partial M_j^2}(x_i,t) &\approx x_{i+e_j}- 2x_{i} + x_{i-e_j}.
\end{aligned}
$$
Since the gradient of the determinant depends on the second derivatives, in total each agent requires information about the velocities of its $2n$ neighbouring agents and the positions of its $2 n^2$ neighbouring agents, including diagonal ones.

Finally, substituting \eqref{Drones_discretizations} into \eqref{Drones_control_continuous}, the formula for the control action $\tau_i$ appears as
\begin{equation} \label{Drones_control}
\begin{aligned}
\tau_i &= \left[W_i G_i^{-1} + \Sum_{j=1}^n [W_i^T]_j \cdot [G_i^{-1}]_j - \alpha \right] v_i +{} \\ &{}+ \left[\beta I - v_i v_i^T\right] \frac{1}{\det G_i} G_i^{-T} \frac{\partial (\det G_i)}{\partial M}^T  +{} \\ & {}+ \det G_i \Big[\alpha \rho_d u_d + (\beta I - u_d u_d^T) \nabla \rho_d^T - \rho_d (\nabla\cdot u_d) u_d \Big].
\end{aligned}
\end{equation}

\subsection{Boundary conditions}

For the system \eqref{Drones_wave} to converge to zero proper boundary conditions should be used. Namely, the continuation should be chosen such that $\rho = 0$ outside of the formation. As it was shown in \ref{SEC_boundary}, boundary conditions for PDE correspond to ''ghost agents'' in the ODE case. In particular, information about neighbour agents is used in \eqref{Drones_discr_G} and \eqref{Drones_discr_W}. Therefore specifying boundary conditions means specifying positions $x_{i \pm e_j}$ and velocities $v_{i \pm e_j}$ for the nonexisting agents.

Assume agent $i-e_j$ is a ghost agent. One natural choice, which we will use for velocities, is to take $v_{i-e_j} = 2 v_i - v_{i+e_j}$. Being substituted in \eqref{Drones_discr_W} this leads to an approximation of the velocity gradient based solely on the $i$ and $i+e_j$ agents.

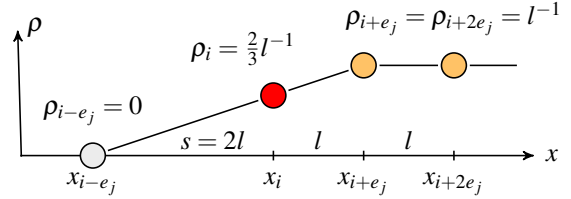
\begin{figure}[h]
\begin{center}
\begin{tikzpicture}[modal]
	
	\path[->] (-0.95,-0.05) edge node[pos=1,right]{$\rho$} (-1.0,1.7);
	\path[->] (-1.0,0) edge node[pos=1,right]{$x$} (5.9,0);
	\path[-] (0,-0.1) edge node[below, yshift=-0.1cm]{$x_{i-e_j}$} (0,0.1);
	\path[-] (2.4,-0.1) edge node[below, yshift=-0.1cm]{$x_{i}$} (2.4,0.1);
	\path[-] (3.6,-0.1) edge node[below, yshift=-0.1cm]{$x_{i+e_j}$} (3.6,0.1);
	\path[-] (4.8,-0.1) edge node[below, yshift=-0.1cm]{$x_{i+2e_j}$} (4.8,0.1);
	
	\node[circle,draw,minimum size=0.1cm,fill=gray!15] at (0,0) (ghost) {};
	\node[circle,draw,minimum size=0.1cm,fill=red] at (2.4,0.8) (x0) {};
	\node[circle,draw,minimum size=0.1cm,fill=red!40!yellow!60] at (3.6,1.2)  (x1){};
	\node[circle,draw,minimum size=0.1cm,fill=red!40!yellow!60] at (4.8,1.2)  (x2){};
	\node at (5.8,1.2)  (x3){};
	
	\path[-] (ghost) edge (x0);
	\path[-] (x0) edge (x1);
	\path[-] (x1) edge (x2);
	\path[-] (x2) edge (x3);
	
	\node at (1.6,0.2) {$s=2l$};
	\node at (3.0,0.2) {$l$};
	\node at (4.2,0.2) {$l$};
	
	\node [above=0.1cm of ghost] {$\rho_{i-e_j} = 0$};
	\node [above=0.1cm of x0, xshift=-0.4cm] {$\rho_{i} = \frac{2}{3}l^{-1}$};
	\node [above=0.1cm of x2] {$\rho_{i+e_j} = \rho_{i+2e_j} = l^{-1}$};
	
\end{tikzpicture}
\end{center}
\vskip -10pt
\caption{Left boundary of the system \eqref{Drones_ODE} with control \eqref{Drones_control}. Agent $i$ is on the boundary, the position of the ''ghost agent'' $i-e_j$ is chosen such that $\rho$ linearly goes to zero at $x_{i-e_j}$.} \label{FIG_boundary_drones}
\end{figure}

This idea can't be used for $x_{i-e_j}$ since in this case the compression tensor \eqref{Drones_discr_G} will not ''feel'' that the drone $i$ is on the border. Instead we will use such an approximation that the density near the border will linearly diminish to zero, see Fig.~\ref{FIG_boundary_drones}. Namely, let us look at 1D case and fix $i$-th agent to be on the left border. Assume further that the distance between each pair of existing agents is constant and equal to $l$. Then $\rho_{i+e_j} = l^{-1}$. Also we define $\rho_{i-e_j} = 0$ such that the ghost agent has its density zero. Define an unknown distance $s := x_i - x_{i-e_j}$. Then asking for a linear dependency of a density on position, we have necessarily 
$$
\rho_i = \frac{l\rho_{i-e_j} + s \rho_{i+e_j}}{l + s} = \frac{s}{l(l+s)}.
$$
But by \eqref{Drones_discr_G} $\rho_i = 2 / (l + s)$, which immediately gives the answer $s = 2l$, or
\begin{equation}
x_{i-e_j} = x_i + 2(x_i - x_{i+e_j}) = 3 x_i - 2 x_{i+e_j}.
\end{equation}
This finalizes the formulation of the boundary conditions and thus the correct implementation of \eqref{Drones_control}.

\subsection{Numerical Simulation}

To demonstrate the control policy \eqref{Drones_control} we performed a numerical simulation of a cubic formation of 512 drones in 3D space. The goal was to reach a cubic formation, fly through a window and restore the cubic formation after the maneuver. 

Assume the center of the window is placed at the point $(x_0, 0, 0)$, and the formation should fly through it starting from the origin. The desired velocity field $u_d(x,y,z)$ able to fulfill the task was constructed as 
$$
u_{d_x} = 1, \quad u_{d_{y|z}} = 0.05 \atan(x-x_0) e^{-\frac{(x-x_0)^2}{100}} y|z,
$$
where $y|z$ denotes $y$ or $z$, see the left panel of Fig.~\ref{FIG_vel_l2} for the streamlines projected on the $x$-$y$ plane. For simplicity the desired system \eqref{Drones_des} was simulated by first-order integrators following the desired velocity profile, and the density $\rho_d(x,t)$ was interpolated between agents.

Both the desired system \eqref{Drones_des} and the real system \eqref{Drones_ODE} were simulated for the cubic formation of $8\times 8\times 8$ drones using Euler method. The initial positions for the real system were multiplied by 2 in comparison to the desired system and a uniform noise $U(-2,2)$ was added. The control gains were chosen as $\alpha = 3$ and $\beta = 100$. The convergence of the real density to the desired one is shown on the right panel of Fig.~\ref{FIG_vel_l2} and snapshots of the simulation are presented in Fig.~\ref{FIG_full}. It is clear that the real formation, being heavily disturbed in the beginning, converges to the desired shape in less than 5 seconds and then follows the desired pattern, successfully passing through the window.

\begin{figure}
\begin{center}
\includegraphics[width=\linewidth]{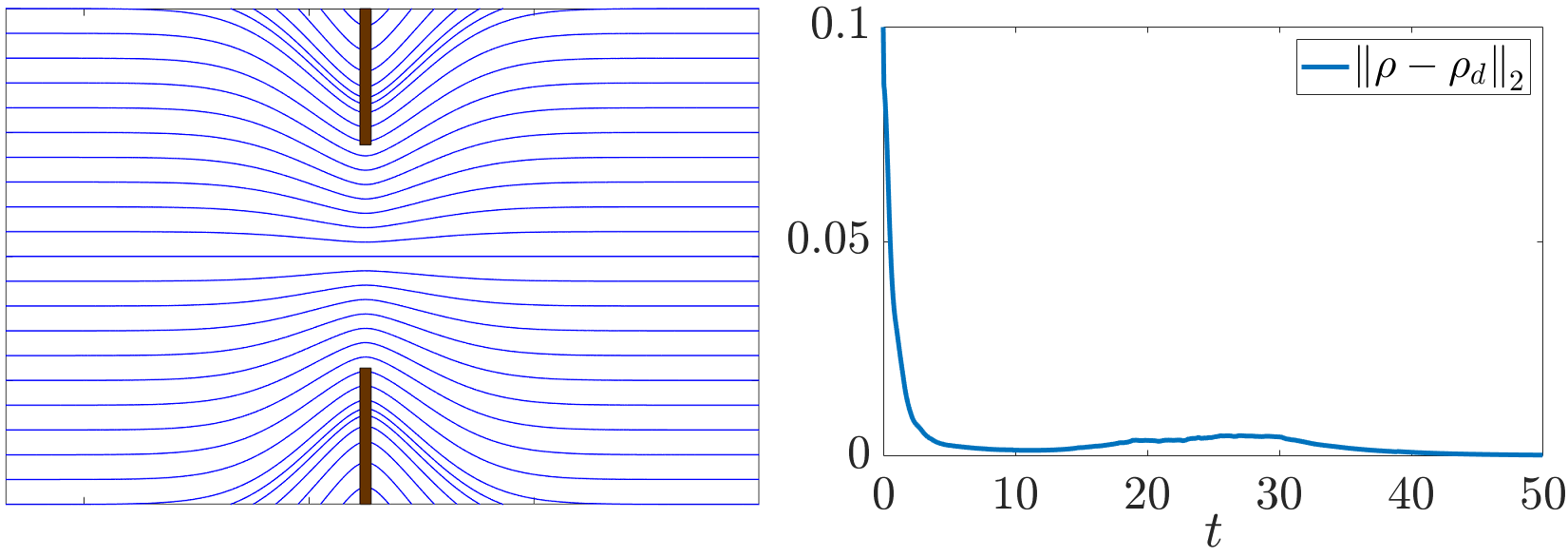}
\caption{Left: streamlines of the desired velocity field $u_d(x,y)$. Right: convergence of the $L_2$ norm of the density deviation.} \label{FIG_vel_l2}
\end{center}
\end{figure}

\begin{figure}
\begin{center}
\includegraphics[width=\linewidth]{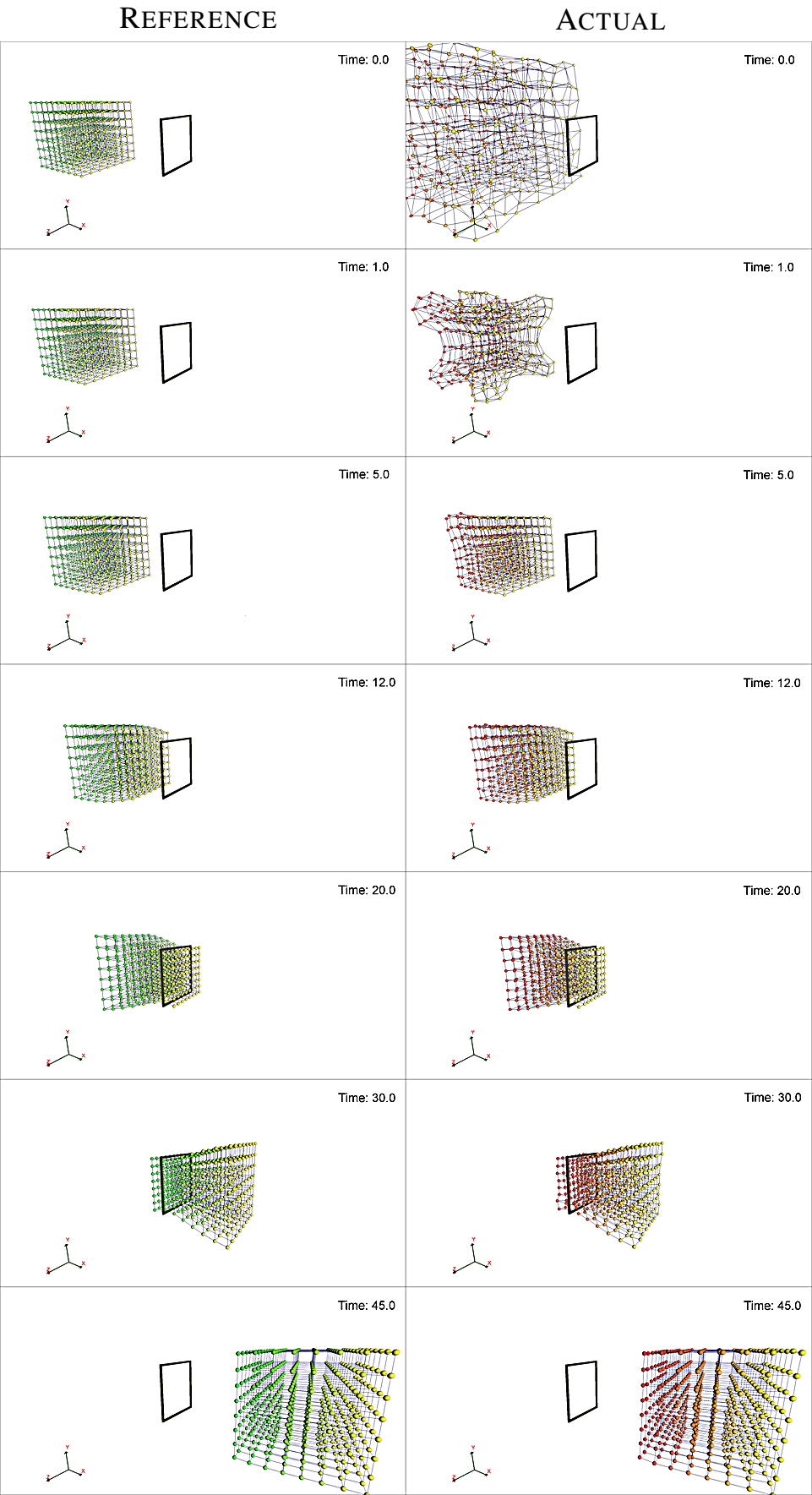}
\caption{Simulation of drones flying through window. Rows correspond to times $t=\{0s,1s,5s,12s,20s,30s,45s\}$. Left column, reference: desired system \eqref{Drones_des}, governed by single integrators. Right column, actual: heavility perturbed real system \eqref{Drones_ODE} with control \eqref{Drones_control} which converges to the desired one.} \label{FIG_full}
\end{center}
\end{figure}

\section{Conclusion}

We presented a general process of transformation of ODE systems into their PDE counterparts, defining the continuation to be valid if the original ODE system could be obtained from the PDE version by a correct discretization. We further showed that the spectrum of PDE converges to the spectrum of ODE. The continuation method was then elaborated for many classes of systems including nonlinear, multidimensional and space- and time-varying. Based on this method, new continuous models can be derived and further utilized for analysis and control purposes.

As an example we used the continuation to show how the Euler equations for compressible fluid can be derived from the newtonian particle interactions, providing more intuition into Hilbert's 6th problem. The same continuation was then used to describe a robot formation flying through window. We developed a control algorithm to stabilize a desired trajectory based on a continuous representation of the formation. This algorithm is distributed as every robot requires information only about neighbouring robots.

It would be desirable to study further the continuation method, namely to provide quantitative measures of how close is a PDE solution compared to the original ODE one.

\appendix

\begin{lemma}
Let $J(x,t) \in \R^{n\times n}$ be the Jacobian matrix of function $M(x,t)$. Let $J(x,t)$ satisfies the dynamic equation
\begin{equation} \label{App_3_J}
\frac{\partial J}{\partial t} = -\frac{\partial (Ju)}{\partial x},
\end{equation} 
where $u = u(x,t)$ is some vector field. Then the determinant $\det J$ satisfies the same equation:
\begin{equation} \label{App_3_detJ}
\frac{\partial \det J}{\partial t} = -\frac{\partial}{\partial x} \cdot \left(\det J \cdot u \right).
\end{equation} 
\end{lemma}
\begin{proof}
First of all let us rewrite \eqref{App_3_J} for one element $J_{ik}$ of the matrix $J$:
\begin{equation} \label{App_3_J_ik}
\begin{aligned}
\frac{\partial J_{ik}}{\partial t}& = -\frac{\partial (J_i \: u)}{\partial x_k} = -\Sum_{j=1}^n \frac{\partial^2 M_i}{\partial x_k \partial x_j} u_j -\Sum_{j=1}^n J_{ij} \frac{\partial u_j}{\partial x_k} = \\ &= -\Sum_{j=1}^n \frac{\partial J_{ik}}{\partial x_j} u_j -\Sum_{j=1}^n J_{ij} \frac{\partial u_j}{\partial x_k},
\end{aligned}
\end{equation} 
where we used the fact that $J = \partial M/ \partial x$.

Now let us recall the definition of the determinant: $\det J = \sum_{\sigma} \sgn(\sigma) \prod_{i=1}^n J_{\sigma_i, i}$, where $\sigma$ is a permutation of the set $\{1,2...n\}$ and $\sum_\sigma$ is taken over all possible permutations, with $\sgn(\sigma)$ being the sign of the permutation. Let us take the time derivative and then substitute \eqref{App_3_J_ik}:
\begin{equation} \label{App_3_det_J_t}
\begin{aligned}
&\frac{\partial \det J}{\partial t} = \Sum_{\sigma} \sgn(\sigma) \Sum_{k=1}^n \frac{\partial J_{\sigma_k, k}}{\partial t} \Prod_{i=1, i\not = k}^n J_{\sigma_i, i} = \\&= -\Sum_{j=1}^n \Sum_{\sigma} \sgn(\sigma) \Sum_{k=1}^n \left[\frac{\partial J_{\sigma_k, k}}{\partial x_j} u_j + J_{\sigma_k,j} \frac{\partial u_j}{\partial x_k} \right] \Prod_{i=1, i\not = k}^n J_{\sigma_i, i}
\end{aligned}
\end{equation}

We will investigate two parts of \eqref{App_3_det_J_t}, corresponding to the first and the second terms inside the square brackets. For the first part we have
\begin{equation} \label{App_3_det_J_p1}
\begin{aligned}
-\Sum_{j=1}^n \Sum_{\sigma} \sgn(\sigma) \Sum_{k=1}^n \frac{\partial J_{\sigma_k, k}}{\partial x_j} u_j \Prod_{i=1, i\not = k}^n J_{\sigma_i, i} &= \\ = -\Sum_{j=1}^n \frac{\partial \det J}{\partial x_j} u_j &= -\frac{\partial \det J}{\partial x} u.
\end{aligned}
\end{equation}

The second part is a little bit more tricky:
\begin{equation*}
\footnotesize{
\begin{aligned}
&-\Sum_{j=1}^n \Sum_{\sigma} \sgn(\sigma) \Sum_{k=1}^n J_{\sigma_k,j} \frac{\partial u_j}{\partial x_k} \Prod_{i=1, i\not = k}^n J_{\sigma_i, i}  = -\det J \Sum_{j=1}^n \frac{\partial u_j}{\partial x_j} -{} \\ &{}-\Sum_{j=1}^n \Sum_{\sigma} \sgn(\sigma) \Sum_{k=1,k\not=j}^n J_{\sigma_k,j} \frac{\partial u_j}{\partial x_k} \Prod_{i=1, i\not = k}^n J_{\sigma_i, i}.
\end{aligned}}
\end{equation*}
Here we split the summation over $k$ into the term with $k=j$ and all other terms. The first one immediately gives the determinant multiplied by the divergence of the vector field. It appears that the sum over all other terms is zero. Indeed, imagine a permutation $\bar\sigma$ such that it is equal to $\sigma$ except $\sigma_j$ and $\sigma_k$ are swapped. Then the sign of $\bar\sigma$ is opposite to the sign of $\sigma$. Further, since the product $J_{\sigma_k,j} J_{\sigma_j,j}$ is the only way in which $\sigma_k$ and $\sigma_j$ enter the formula, the absolute value does not change with the change of permutation. Therefore for each $j,k$ and for each permutation there exists a permutation which cancels them out. 

Finally, substitution of the nonzero term of the last equation and \eqref{App_3_det_J_p1} into \eqref{App_3_det_J_t} leads to \eqref{App_3_detJ}.
\end{proof}

\begin{lemma}
Let $\partial x / \partial M$ be isotropic, i.e. represented by a scalar multiplied by a rotation matrix, and let $\rho~=~\det(\partial M/\partial x)$. Then
\begin{equation} \label{App_4}
\nabla\cdot\left(\rho\frac{\partial x}{\partial M}\right)\frac{\partial x}{\partial M}^T = 0.
\end{equation}
\end{lemma}
\begin{proof}
Define $\lambda = \lambda(\partial M/\partial x)$, thus $\rho = \lambda^n$. By isotropy
\begin{equation} \label{App_4_isotropy}
\frac{\partial x}{\partial M} = \lambda^{-2} \frac{\partial M}{\partial x}^T, 
\end{equation}
therefore the left-hand side of \eqref{App_4} is
\begin{equation} \label{App_4_Mx}
\begin{aligned}
&\nabla\cdot\left(\lambda^{n-2}\frac{\partial M}{\partial x}^T\right)\frac{\partial M}{\partial x} \lambda^{-2} = \\ &= \lambda^{n-4} \nabla\cdot\left(\frac{\partial M}{\partial x}^T\right)\frac{\partial M}{\partial x} + (n-2) \lambda^{n-3} \frac{\partial \lambda}{\partial x},
\end{aligned}
\end{equation}
where we used \eqref{Prop_1_scalar} and \eqref{App_4_isotropy}. Now let us investigate the first term more closely. Taking the divergence and looking at $j$-th element, we see that
\begin{equation} \label{App_4_Mx_div}
\left[\nabla\cdot\left(\frac{\partial M}{\partial x}^T\right)\frac{\partial M}{\partial x}\right]_j = \Sum_{k=1}^n \frac{\partial^2 M}{\partial x_k^2}^T \frac{\partial M}{\partial x_j}.
\end{equation} 
Now, By isotropy
\begin{equation} \label{App_4_tech_1}
\frac{\partial M}{\partial x_j}^T \frac{\partial M}{\partial x_k} = 0 \;\; \forall j \not= k, \quad \frac{\partial M}{\partial x_k}^T \frac{\partial M}{\partial x_k} = \lambda^2.
\end{equation}
Taking the derivative of the multiplication of basis vectors:
\begin{equation} \label{App_4_tech_2}
\frac{\partial}{\partial x_j} \left( \frac{\partial M}{\partial x_k}^T \frac{\partial M}{\partial x_k} \right) = 2 \frac{\partial^2 M}{\partial x_j \partial x_k}^T \frac{\partial M}{\partial x_k},
\end{equation}
but at the same time the value under the derivative is $\lambda^2$ by \eqref{App_4_tech_1}, therefore
\begin{equation} \label{App_4_tech_3}
\frac{\partial}{\partial x_j} \left( \frac{\partial M}{\partial x_k}^T \frac{\partial M}{\partial x_k} \right) = \frac{\partial \lambda^2}{\partial x_j} = 2 \lambda \frac{\partial \lambda}{\partial x_j}.
\end{equation}
Then, taking the derivative of multiplication of different basis vectors with $j \not= k$, by \eqref{App_4_tech_1} we obtain zero:
\begin{equation*}
\frac{\partial}{\partial x_k} \left( \frac{\partial M}{\partial x_j}^T \frac{\partial M}{\partial x_k} \right) = \frac{\partial^2 M}{\partial x_j \partial x_k}^T \frac{\partial M}{\partial x_k} + \frac{\partial M}{\partial x_j}^T \frac{\partial^2 M}{\partial x_k^2} = 0,
\end{equation*}
which by equality of \eqref{App_4_tech_2} and \eqref{App_4_tech_3} means that for $j \not= k$
\begin{equation} \label{App_4_tech_4}
\frac{\partial M}{\partial x_j}^T \frac{\partial^2 M}{\partial x_k^2} = -\frac{\partial^2 M}{\partial x_j \partial x_k}^T \frac{\partial M}{\partial x_k} = -\lambda \frac{\partial \lambda}{\partial x_j}.
\end{equation}
In the case of $j=k$ by equality of \eqref{App_4_tech_2} and \eqref{App_4_tech_3} we have
\begin{equation} \label{App_4_tech_5}
\frac{\partial^2 M}{\partial x_j^2}^T \frac{\partial M}{\partial x_j} = \lambda \frac{\partial \lambda}{\partial x_j}.
\end{equation}

Combination of \eqref{App_4_tech_4} and \eqref{App_4_tech_5} means that \eqref{App_4_Mx_div} is
\begin{equation} \label{App_4_Mx_div_2}
\left[\nabla\cdot\left(\frac{\partial M}{\partial x}^T\right)\frac{\partial M}{\partial x}\right]_j = (2-n) \lambda \frac{\partial \lambda}{\partial x_j}.
\end{equation} 
Finally, substituting \eqref{App_4_Mx_div_2} in \eqref{App_4_Mx} gives zero.
\end{proof}

\section*{Acknowledgment}

The Scale-FreeBack project has received funding from the European Research Council (ERC) under the European Union’s Horizon 2020 research and innovation programme (grant agreement N 694209).

\ifCLASSOPTIONcaptionsoff
  \newpage
\fi




\begin{thebibliography}{99}

\bibitem{EULER} L. Euler. Principia motus fluidorum. Novi commentarii academiae scientiarum Petropolitanae, 271-311. 1761.

\bibitem{MAXWELL} J. C. Maxwell. A treatise on electricity and magnetism (Vol. 1). Clarendon press. 1873.

\bibitem{DISCRETIZATION} J. G. Verwer and  J. M. Sanz-Serna. Convergence of method of lines approximations to partial differential equations. Computing, 33(3-4), 297-313. 1984.

\bibitem{CHOICE_DISCRETIZATION} W. F. Ames. Numerical methods for partial differential equations. Academic press. 2014.

\bibitem{AOKI68} M. Aoki. Control of large-scale dynamic systems by aggregation. IEEE Transactions on Automatic Control, vol. 13, no. 3, pp. 246–253, 1968.

\bibitem{UMAR19} M. U. B. Niazi, X. Cheng, C. Canudas-de-Wit and J. Scherpen. Structure-based Clustering Algorithm for Model Reduction of Large-scale Network Systems. CDC 2019 - 58th IEEE Conference on Decision and Control, Nice, France, Dec 2019.

\bibitem{UMAR20} M. U. B. Niazi, D. Deplano, C. Canudas-de-Wit, and A. Y. Kibangou. Scale-free estimation of the average state in large-scale systems. IEEE Control Systems Letters, vol. 4, no. 1, pp. 211–216, Jan 2020.

\bibitem{NIK20} D. Nikitin, C. Canudas-de-Wit and P. Frasca. Control of Average and Deviation in Large-Scale Linear Networks. Submitted to IEEE Transactions on Automatic Control.

\bibitem{ACEBRON05} J. A. Acebron, L. L. Bonilla, C. J. P. Vicente, F. Ritort and R. Spigler. The Kuramoto model: A simple paradigm for synchronization phenomena. Reviews of modern physics, 77(1), 137. 2005.

\bibitem{POPULATION00} D. Q. Nykamp and D. Tranchina. A population density approach that facilitates large-scale modeling of neural networks: Analysis and an application to orientation tuning. Journal of computational neuroscience, 8(1), 19-50. 2000.

\bibitem{GRA82} H. Grabert. Projection operator techniques in nonequilibrium statistical mechanics. Springer Tracts in Modern Physics, vol. 95, Springer-Verlag, Berlin-New York, 1982.

\bibitem{MOMENTS14} Y. Yang, D. V. Dimarogonas and X. Hu. Shaping up crowd of agents through controlling their statistical moments. 2015 European Control Conference (ECC), pp. 1017-1022, Linz, 2015.

\bibitem{MOMENTS20} S. Zhang, A. Ringh, X. Hu and J. Karlsson. Modeling collective behaviors: A moment-based approach. IEEE Transactions on Automatic Control, 2020.

\bibitem{MOMENTS_CL} C. Kuehn. Moment closure—a brief review. In Control of self-organizing nonlinear systems (pp. 253-271). Springer, Cham. 2016.

\bibitem{MY_IFAC} D.~Nikitin, C.~Canudas-de-Wit, P.~Frasca. Shape-based nonlinear model reduction for 1D conservation laws. 21st IFAC World Congress 2020, Berlin, Germany, July 11-17, 2020.

\bibitem{GRAPHONS1} D. Glasscock. a Graphon?. Notices of the AMS, 62(1). 2015.

\bibitem{GRAPHONS2} L. Lovász. Large networks and graph limits (Vol. 60). American Mathematical Soc.. 2012.

\bibitem{GRAPHONS3} S. Gao and P. E. Caines. Graphon Control of Large-scale Networks of Linear Systems. in IEEE Transactions on Automatic Control, 2019.

\bibitem{GRAPHONS4} R. Vizuete, P. Frasca and F. Garin. Graphon-based sensitivity analysis of SIS epidemics. IEEE Control Systems Letters, IEEE, 2020, 4 (3), pp.542 - 547.

\bibitem{COMPUTATIONAL_GRAPH}  A. G. Baydin, B. A. Pearlmutter, A. A. Radul and J. M. Siskind. Automatic differentiation in machine learning: a survey. The Journal of Machine Learning Research, 18(1), 5595-5637. 2017.


\bibitem{BASSAM} B. Bamieh, F. Paganini and M. A. Dahleh. Distributed control of spatially invariant systems. IEEE Transactions on automatic control, 47(7), 1091-1107. 2002.

\bibitem{LAURENT} A. E. Frazho and W. Bhosri. Toeplitz and Laurent Operators. In: An Operator Perspective on Signals and Systems. Operator Theory: Advances and Applications (Linear Operators and Linear Systems), vol 204. Birkhäuser Basel. 2010. 

\bibitem{KHOLMOGOROV} A. N. Kolmogorov. On the representation of continuous functions of several variables by superposition of continuous functions of one variable and addition. Dokl. Akad. Nauk SSSR 114, 953-956. 1957.


\bibitem{EE_MOSKOWITZ} G. F. Newell. A simplified theory of kinematic waves in highway traffic, part I: General theory. Transportation Research Part B: Methodological, 27(4), 281-287. 1993.

\bibitem{EE_BOOK_BZ} I. Gallagher, L. Saint-Raymond and B. Texier. From Newton to Boltzmann: hard spheres and short-range potentials. European Mathematical Society. 2013.

\bibitem{EE_BOOK_EXP} L. Saint-Raymond. Hydrodynamic limits of the Boltzmann equation. Lecture Notes in Mathematics, vol. 1971. Springer-Verlag, 2009.

\bibitem{EE_CHAPMAN} S. Chapman and T. G. Cowling. The Mathematical Theory of Non-Uniform Gases. Cambridge University Press, Cambridge, 1970.

\bibitem{EE_GRAD} H. Grad. On the kinetic theory of rarefied gases. Commun. Pure Appl. Math. 2, 325, 1949.

\bibitem{EE_REVIEW} A. N. Gorban and I. Karlin. Hilbert's 6th Problem: exact and approximate hydrodynamic manifolds for kinetic equations. Bull. Amer. Math. Soc. 51 (2): 186–246, 2014. 

\bibitem{EE_VLASOV_EULER} S. Caprino, R. Esposito, R. Marra and M. Pulvirenti. Hydrodynamic limits of the Vlasov equation. Communications in partial differential equations, 18(5-6), 805-820. 1993.


\bibitem{EE_SPHERE} R. Takloo-Bighash. How many lattice points are there on a circle or a sphere?. In: A Pythagorean Introduction to Number Theory. Undergraduate Texts in Mathematics. Springer, Cham. 2018. 




\bibitem{RS_REVIEW15} K. K. Oh, M. C. Park and H. S. Ahn. A survey of multi-agent formation control. Automatica, 53, 424-440. 2015.

\bibitem{RS_REVIEW18} S. Chung, A. A. Paranjape, P. Dames, S. Shen and V. Kumar. A Survey on Aerial Swarm Robotics. IEEE Transactions on Robotics, vol. 34, no. 4, pp. 837-855, Aug. 2018.

\bibitem{RS_BIRDS95} J. Toner and Y. Tu, Long-range order in a two-dimensional dynamical XY model: How birds fly together. Phys. Rev. Lett., vol. 75, no. 3, pp. 4326–4329, 1995.

\bibitem{RS_KRSTIC} J. Qi, R. Vazquez and M. Krstic. Multi-agent deployment in 3-D via PDE control. IEEE Transactions on Automatic Control, 60(4), 891-906. 2014.

\bibitem{RS_PdE06} G. Ferrari-Trecate, A. Buffa, and M. Gati. Analysis of coordination in multi-agent systems through partial difference equations. IEEE Trans. Autom. Control, vol. 51, no. 6, pp. 1058–1063, Jun. 2006.

\bibitem{RS_PdE08} M. Ji, G. Ferrari-Trecate, M. Egerstedt, and A. Buffa. Containment control in mobile networks. IEEE Trans. Autom. Control, vol. 53, no. 8, pp. 1972–1975, Sep. 2008.

\bibitem{WAVE} G. B. Folland. Introduction to partial differential equations (Vol. 102). Princeton university press. 1995.

\end{thebibliography}
%

%


\begin{IEEEbiography}[{\includegraphics[width=1in,height=1.25in,clip,keepaspectratio]{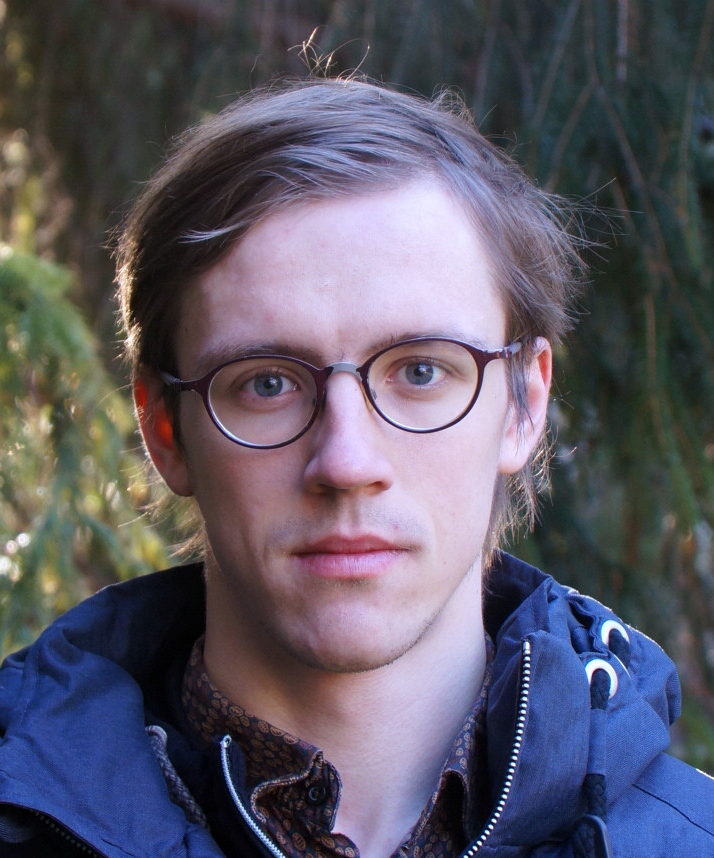}}]
{Denis Nikitin}
received both the B.Sc. and the M.Sc. degree in Mathematics and Mechanics Faculty of Saint Petersburg State University, Saint Petersburg, Russia, specializing on the control theory and cybernetics. He won several international robotics competitions while being student and was a teacher of robotics in the Math{\&}Phys Lyceum 239 in Saint Petersburg.

He is currently a doctoral researcher at CNRS, GIPSA-Lab, Grenoble, France. His current research mainly focuses on control of large-scale systems.
\end{IEEEbiography}


\begin{IEEEbiography}[{\includegraphics[width=1in,height=1.25in,clip,keepaspectratio]{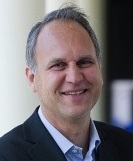}}]{Carlos Canudas-de-Wit}
(F'16) was born in Villahermosa, Mexico, in 1958. He received the B.S. degree in electronics and communications from the Monterrey Institute of Technology and Higher Education, Monterrey, Mexico, in 1980, and the M.S. and Ph.D. degrees in automatic control from the Department of Automatic Control, Grenoble Institute of Technology, Grenoble, France, in 1984 and 1987, respectively. He is currently a Directeur de recherche (Senior Researcher) with CNRS, Grenoble, where he is the Leader of the NeCS Team, a joint team of GIPSA-Lab (CNRS) and INRIA, on networked controlled systems. 

Dr. Canudas-de-Wit is an IFAC Fellow. He was an Associate Editor of the IEEE Transactions on Automatic Control, the Automatica, the IEEE Transactions on Control Systems Technology. He is an Associate Editor of the Asian Journal of Control, and the IEEE Transactions on Control of Network Systems. He served as the President of the European Control Association from 2013 to 2015, and a member of the IEEE Board of Governors of the Control System Society from 2011 to 2014. He holds the ERC Advanced Grant Scale-FreeBack from 2016 to 2021. 
\end{IEEEbiography}



\begin{IEEEbiography}[{\includegraphics[width=1in,height=1.25in,clip,keepaspectratio]{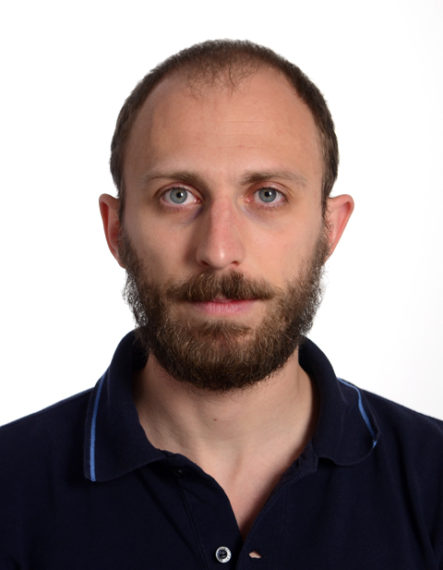}}]{Paolo Frasca}
(M’13–SM’18) received the Ph.D. degree from Politecnico di Torino, Turin, Italy, in 2009. After Postdoctoral appointments with the CNR-IAC, Rome, and in Torino, he has been an Assistant Professor with the University of Twente, Enschede, The Netherlands, from 2013 to 2016. Since October 2016, he has been a CNRS Researcher with GIPSA-lab, Grenoble, France. 

His research interests include the theory of network systems and cyber-physical systems, with applications to robotics, sensors, infrastructural, and social networks. On these topics, he has (co)authored more than 60 journal and conference papers and the book Introduction to Averaging Dynamics Over Networks (Springer). He has been an Associate Editor for the Editorial Boards of several conferences and journals, including the IEEE Control Systems Letters.
\end{IEEEbiography}




\end{document}